\documentclass{article}

\usepackage{enumitem}
\usepackage{amssymb}
\usepackage{amsmath}
\usepackage{amsthm}
\usepackage{amsopn}
\usepackage{graphicx}
\usepackage{bbm}
\usepackage{mathdots}
\usepackage[export]{adjustbox}
\usepackage{mathdots}
\usepackage{hyperref}
\usepackage{float}

\hypersetup{
    colorlinks = true,
    citecolor=black,
    filecolor=black,
    linkcolor=black,
    urlcolor=black
    linktoc=all,     
    linkcolor=blue,  
}

\newtheorem{theorem}{Theorem}[section]
\newtheorem{definition}[theorem]{Definition}
\newtheorem{claim}[theorem]{Claim}
\newtheorem{lemma}[theorem]{Lemma}
\newtheorem{conclusion}[theorem]{Conclusion}

\newtheorem{conjecture}[theorem]{Conjecture}

\newtheorem*{conjecturee}{Conjecture}

\newcommand{\CC}{\mathbb C}

\newcommand{\FF}{\mathcal F}

\newcommand{\RR}{\mathbb R}

\newcommand{\NN}{\mathbb N}
\newcommand{\ZZ}{\mathbb Z}

\newcommand{\HH}{\mathcal H}
\newcommand{\TT}{\mathbb T}

\newcommand{\heaviside}{\mathbbm{1}_{(0,\infty)}}
\newcommand{\LL}{\mathcal L}

\DeclareMathOperator{\sgn}{sgn}

\DeclareMathOperator{\op}{op}
\DeclareMathOperator{\Id}{Id}

\graphicspath{{images/}}
\title{Commutators of spectral projections of spin operators}
\author{\makebox[.9\textwidth]{Ood Shabtai}}
\begin{document}
\maketitle
\begin{abstract} We present a proof that the operator norm of the commutator of certain spectral projections associated with spin operators converges to $\frac 1 2$ in the semiclassical limit. The ranges of the projections are spanned by all eigenvectors corresponding to positive eigenvalues. The proof involves the theory of Hankel operators on the Hardy space. A discussion of several analogous results is also included, with an emphasis on the case of finite Heisenberg groups.\end{abstract}
\tableofcontents
\section{Introduction}
Let $J_x, J_y, J_z$ denote the generators of an irreducible, unitary, $n$-dimensional representation of $SU(2)$, satisfying the commutation\footnote{Here, $[A,B] = AB - BA$ denotes the commutator of a pair of linear operators $A,B$.} relations
\begin{equation*} [J_x, J_y] = i J_z,\ [J_y, J_z] = i J_x,\ [J_z, J_x] = i J_y. \end{equation*}
Consider the commutator
\begin{equation*} C_n = \left[\mathbbm{1}_{(0,\infty)}(J_x),\mathbbm{1}_{(0,\infty)}(J_z)\right],\end{equation*}
where $\mathbbm{1}_{(0,\infty)}$ denotes the indicator function of $(0,\infty) \subset \RR$.  The main results of the present work are that
\begin{theorem}[L. Polterovich]\label{2mod4} $\Vert C_{4n+2} \Vert_{\op} = \frac 1 2$ for every $n \in \NN$,\end{theorem}
and
\begin{theorem}\label{maintheorem}$\lim_{n \to \infty} \Vert C_n \Vert_{\op} = \frac 1 2$.\end{theorem}

The sequence $(\Vert C_n\Vert_{\op} )_{n=2}^\infty$ is bounded from above by $\frac 1 2$ due to a general fact about commutators of orthogonal projections (\cite{qihuili}). Nonetheless, it is perhaps not evident a priori that the sequence should converge at all, let alone to the largest possible value.

As it turns out, however, analogues of Theorem \ref{maintheorem} hold for several other families of pairs of spectral projections arising from non-commuting observables. A few such examples are formulated in Section \ref{analoguessubsection}, and a modest extension of Theorem \ref{maintheorem} is included in Section \ref{extensionsubsecsion}. Ultimately, we suspect that the various results presented here are instances of a rather general phenomenon. We refer the reader to Section \ref{morecases} for details and remarks along these lines.
\subsection{Numerical simulations and preliminary results}
The numerical simulations (originally by Y. Le Floch) of $\left(\Vert C_n \Vert_{\op} \right)_{n=2}^\infty$ imply further intriguing properties.
\begin{figure}[H]
    \centering
        \includegraphics[width=1.3\textwidth,center]{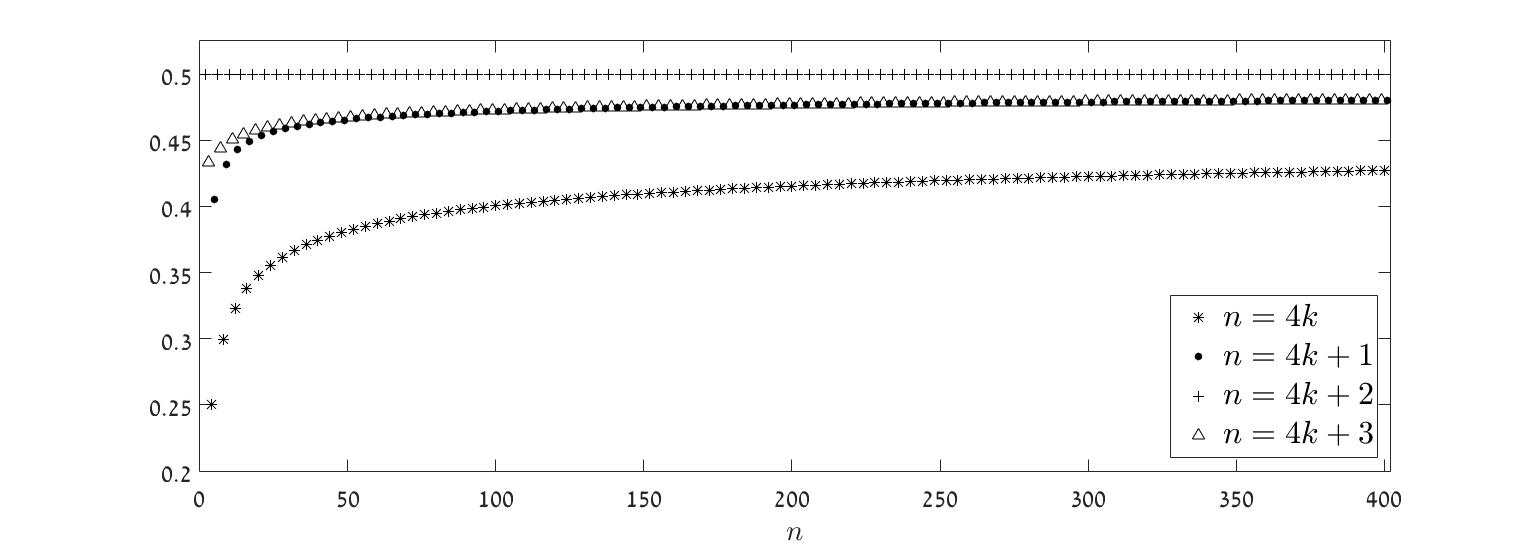}
        \caption{$\Vert C_n \Vert_{\op}$ as a function of $n$.}\label{su2}
\end{figure}
\begin{figure}[H]
    \centering
        \includegraphics[width=1.3\textwidth,center]{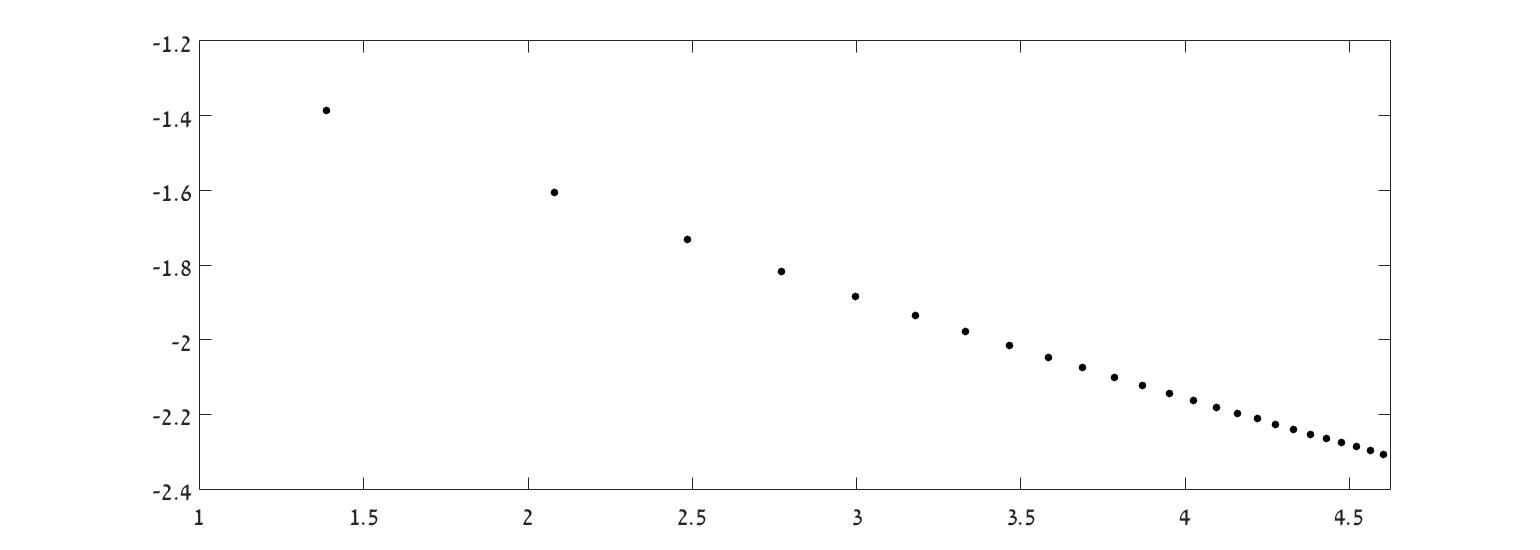} 
        \caption{$\ln\left(\frac 1 2 - \Vert C_{4k} \Vert_{\op} \right)$ as a function of $\ln(4k)$.}
\end{figure}
Notably, $\Vert C_n \Vert_{\op}$ appears to depend on the dimension of the representation modulo $4$. More precisely, 
\begin{conjecture} $\frac{\Vert C_{4n+3} \Vert_{\op} - \Vert C_{4n+1} \Vert_{\op}}{\Vert C_{4n+3} \Vert_{\op} - \Vert C_{4n}\Vert_{\op}} = o(1)$ and $\frac{\Vert C_{4n+3} \Vert_{\op} - \Vert C_{4n+1} \Vert_{\op}}{\frac 1 2 - \Vert C_{4n+p} \Vert_{\op}} = o(1)$ for $p = 0, 1, 3$.\end{conjecture}
However, the convergence rate of the various sequences is presently unknown\footnote{The application of linear regression suggests a rate not faster than $O\left(n^{-\frac 1 4}\right)$, but otherwise does not seem to provide a clear answer.}.
Finally, it holds that $\Vert C_n \Vert_{\op} \ge \frac 1 4$ for every $n \ge 2$, as depicted above. The lower bound was established as part of our initial studies of the sequence $(C_n)_{n =2}^\infty$, joint with Y. Le Floch and L. Polterovich.

The proof that $\Vert C_{4n+2} \Vert_{\op}= \frac 1 2$ for every $n \in \NN$ is quite straightforward and requires little but symmetries. Otherwise, the lower bound is derived through specific elements of the matrix representing $C_n$ in some orthonormal eigenbasis of $J_z$. The arguments involved are relatively simple when $n$ is odd, thanks to symmetries again. When $n \equiv 0 \mod 4$, however, the bound is established through the use of brute-force techniques. This amounts to cumbersome calculations, since the immediate formula for the elements of $C_n$ is complicated to estimate directly as $n \to \infty$ (and involves sums of products of certain special functions).
\subsection{Outline of the proof of Theorem \ref{maintheorem}}
The proof of Theorem \ref{maintheorem} also revolves around the matrix elements of $C_n$ and their limits as $n \to \infty$. The proof appears in Section \ref{mainproofsection}, relying on two main ingredients which are presented subsequently. Namely, in Section \ref{integralformulasection} we derive a concise integral formula for the elements of $C_n$ in terms of the corresponding matrix elements of the one-parameter subgroup
\begin{equation*} \left\{\left. e^{-i \theta J_y}\right| -2\pi \le \theta < 2\pi \right\}.\end{equation*} Thus, we avoid the initial, complicated expressions for the elements of $C_n$ altogether. In Section \ref{asymptoticsection}, we use a classical (\cite{szego}, 8.21.12) asymptotic estimate for the elements (up to normalization\footnote{The asymptotic estimate is formulated in terms of Jacobi polynomias, rather than Wigner d-functions, and applies in significantly more general settings. See (\ref{szegotheorem}) for the original result, and Conclusion \ref{wignerdasymptotics} for the adaptation to our present settings.}) of $e^{-i\theta J_y}$ to study the aforementioned integral formula. The estimate is valid in a restricted range of indices, which nonetheless suffices for our purpose, i.e., to calculate the limits of sufficiently many of the elements of $C_n$.

Ultimately, for every positive integer $N$ and every $n > 2N+1$, we invert the rows of some sub-matrix of $-C_n$ to obtain a collection $C_{n,N} \in M_N(\CC)$, such that $C_{n,N}$ is a sub-matrix of $C_{n,N+1}$ whenever $n > 2N+3$, and furthermore, $\lim_{n \to \infty} C_{n,N}$ exists and has constant elements along the anti-diagonals.

The latter, we recall, is the defining property of (possibly infinite) Hankel matrices (\cite{pellerbig}). More generally, any operator on some Hilbert space whose matrix relative to some orthonormal basis is a Hankel matrix may be considered as a Hankel operator. Thus, we show that
\begin{equation*} \left(\lim_{n\to \infty} C_{n,N}\right)_{N \ge 1}\end{equation*}
is the sequence of truncated matrices of some fixed Hankel operator\footnote{Notably, the same Hankel operator appears in all of the cases addressed in Section \ref{morecases}.} whose norm, it turns out, equals $\frac 1 2$.
This suffices to conclude the proof.

We note (again) that a slightly extended version of Theorem \ref{maintheorem} is also included in Section \ref{extensionsubsecsion}. The proof is pretty much identical, and involves slightly modified Hankel operators.
\subsection{Analogues of Theorem \ref{maintheorem}}\label{analoguessubsection}
The results of this part are proven in Section \ref{miscproofs}. Our first example involves the standard quantum model for a particle in a line (i.e., in $\RR$). The next two examples are similar to the first, but involve the configuration spaces $\TT, \ZZ_n = \ZZ/n \ZZ$ rather than $\RR$. The last example is formulated in terms of the representation theory of the group of orientation preserving Euclidean plane isometries.

Let $X, \Xi$ denote the position and momentum operators on $L^2(\RR)$, acting on a smooth $f \in L^2(\RR)$ by
\begin{equation*} Xf (x) = xf(x),\ \Xi f(x) = -i\hbar f'(x).\end{equation*}
\begin{theorem}\label{linecommutator} Consider the commutator $C^{(1)}_\hbar = \left[\Pi_X, \Pi_{\Xi}\right]$, where
\begin{equation*} \Pi_X = \mathbbm{1}_{(0,\infty)}(X),\ \Pi_\Xi = \mathbbm{1}_{(0,\infty)}(\Xi).\end{equation*}
Then $C^{(1)}_\hbar \equiv C^{(1)}$ is independent of $\hbar$, and $\Vert C^{(1)} \Vert_{op} = \frac 1 2$. \end{theorem}
Similarly, define the operators $\Theta, Z$ on $L^2(S^1)\simeq L^2\left([0,2\pi),\frac 1 {2\pi} d\theta\right)$ by
\begin{equation*} \Theta u(\theta) = \theta u (\theta),\ Zu(\theta) = -i \frac{2\pi} n u'(\theta),\end{equation*}
where $u\in C^\infty(S^1)$ and $n \in \NN$. The operators $\Theta,\ Z$ may be used to construct an analogue of Weyl quantization for the cylinder $T^* S^1$ (\cite{mukunda}, \cite{cylinder1}, \cite{cylinder2}).
\begin{theorem}\label{ringcommutator} Let $C^{(2)}_{n} = \left[\Pi_\Theta, \Pi_Z \right]$, where
\begin{equation*} \Pi_\Theta = \mathbbm{1}_{(0,\infty)}(\cos \Theta),\ \Pi_Z = \mathbbm{1}_{(0,\infty)}(\cos Z).\end{equation*}
Then $\lim_{n \to \infty} \Vert C^{(2)}_n \Vert_{\op} = \frac 1 2$. The same result holds if we replace the Heaviside function with $\mathbbm{1}_{(a, \infty)}$, where $a \in [0,1)$.\end{theorem}

The previous two examples may be formulated in terms of the representation theory of the Heisenberg groups $H(\RR)$ and $H(\TT)$ (i.e.,  the group of unitary operators on $L^2(\TT)$ generated by translation operators and by operators of multiplication by characters). A similar result holds for representations of finite Heisenberg groups associated to $\ZZ_n$ as $n \to \infty$ (\cite{schwinger, prasad, vourdas, varadaraweisbart}). This problem was suggested to us by D. Kazhdan, and its solution provided the model of the proof for $SU(2)$ (as well as for the rest of the examples).

Let $g_1, g_2$ define an irreducible unitary representation of the finite Heisenberg group $H(\ZZ_n)$ on $l^2(\ZZ_n)$ by
\begin{equation*} g_1(f)(x) = f(x+1),\ g_2(f)(x) = e^{\frac{2\pi x i} n} f(x).\end{equation*}
Let $\Pi_1, \Pi_2$ denote the orthogonal projections on the subspaces of $l^2(\ZZ_n)$ spanned by eigenvectors of $g_1, g_2$ corresponding to eigenvalues with positive real part. Consider the commutator
\begin{equation*} C^{(3)}_n = \left[\Pi_1, \Pi_2 \right].\end{equation*}
The parallel of Theorem \ref{maintheorem} is the following.
\begin{theorem}\label{heisclaim} $\lim_{n\to \infty} \Vert C^{(3)}_n \Vert_{\op} = \frac 1 2$, and the same holds if we replace the Heaviside function with $\mathbbm{1}_{(a, \infty)}$, where $a \in [0,1)$.\end{theorem}
Additionally, as in the case of $SU(2)$,
\begin{theorem}[Y. Le Floch] $\Vert C^{(3)}_{4n+2} \Vert_{\op} \equiv \frac 1 2$ for every $n \in \NN$. \end{theorem}
The numerical simulations of the sequence $\left(\Vert C^{(3)}_n \Vert_{\op}\right)_{n=2}^\infty$ feature some striking similarities with the equivalent simulations (\ref{su2}) for $SU(2)$, including the conjectured dependence on $n \mod 4$.
\begin{figure}[H]
    \centering
        \includegraphics[width=1.3\textwidth,center]{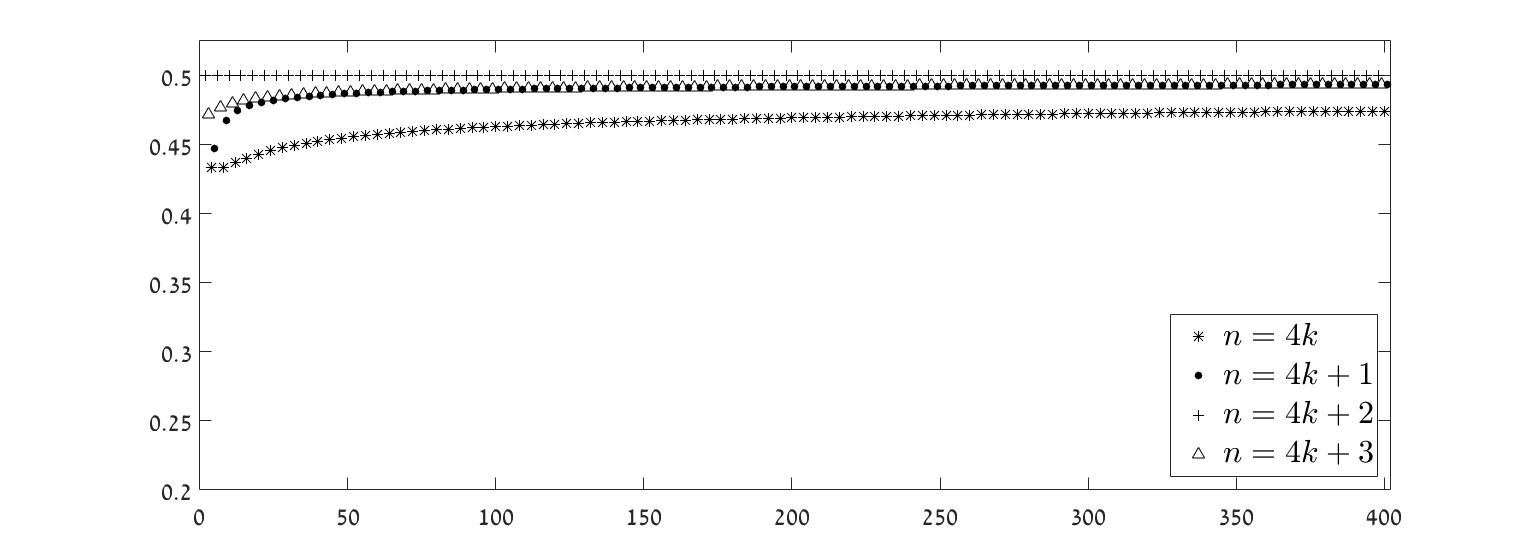} 
        \caption{The norm of $C^{(3)}_{n}$ as a function of $n$ for the Heisenberg groups $H(\ZZ_n)$. Note the similarity to the graph for $SU(2)$. In particular, the graph also appears to depend on $n\mod 4$.}\label{heis}
\end{figure}
Our final example may be derived as a consequence of Theorem \ref{maintheorem}. Let $SE(2)$ denote the group of orientation preserving Euclidean plane isometries. The asymptotic formula \ref{wignerdasymptotics}, which underlies the proof of the Theorem \ref{maintheorem}, provides a non-trivial relation (\cite{inonuwigner, rowedeguisesanders, subagbaruch}) between the representations of $SU(2)$ and of $SE(2)$. This fact led us to study the analogue of $C_n$ for the irreducible representations of $SE(2)$.

Consider $L^2(S^1)\simeq L^2\left([0,2\pi), \frac 1 {2\pi} d\phi \right)$ as before, and fix $R > 0$. Let $X_1,X_2$ denote the multiplication operators $\mathcal M_{R \cos \phi},\ \mathcal M_{R \sin \phi}$ respectively, and let $\Phi$ denote the differentiation operator $f \mapsto -i f'$. Then $X_1,X_2,\Phi$ satisfy the commutation relations
\begin{equation*} [X_1,X_2] = 0,\ [X_2,\Phi] = iX_1,\ [\Phi, X_1] = iX_2,\end{equation*}
hence they generate a unitary representation of $SE(2)$ on $L^2(S^1)$, which is irreducible since $X_1 \pm iX_2 = \mathcal M_{Re^{\pm i \phi}}$ act as raising and lowering operators on the standard basis of $L^2(S^1)$. In fact, every non-trivial irreducible unitary representation of $SE(2)$ is equivalent to the representation generated by $X_1,X_2, \Phi$ for some $R>0$ (\cite{ito}).
\begin{theorem}\label{euclideancommutator} Let $C^{(4)}_R = \left[\Pi_{X_1}, \Pi_\Phi \right]$, where
\begin{equation*} \Pi_{X_1} = \mathbbm{1}_{(0,\infty)}(X_1),\ \Pi_\Phi = \mathbbm{1}_{(0,\infty)}(\Phi).\end{equation*} Then $C^{(4)}_R \equiv C^{(4)}$ is independent of $R$, and $\Vert C^{(4)} \Vert_{\op} = \frac 1 2$. \end{theorem}
\section{Preliminaries on Hankel operators}\label{hankelssubsection}
Theorem \ref{maintheorem} is essentially reduced to the problem of the calculation of $\Vert H_E \Vert_{\op}$ for some Hankel operator $H_E$ which we now specify. The operator $H_E$ appears and plays roughly the same role in all of the results of Section \ref{morecases} as well.

Let $\TT \subset \CC$ denote the unit circle, and declare the functions $z \mapsto z^p,\ p\in \ZZ$ to be an orthonormal basis of $L^2(\TT)$. Let $\Pi_{\TT} : L^2(\TT) \to L^2(\TT)$ denote the Cauchy-Szeg\"{o} projection on the Hardy space
\begin{equation*} H^2(\TT) = \{f\in L^2(\TT)\ |\ \hat f(p) = 0\ \text{for every } p < 0 \}. \end{equation*}
Here, $\hat f(p) = \langle f, z^p\rangle_{L^2(\TT)}$ denotes the $p$-th Fourier coefficient of $f$. Finally, let $\mathcal M_\phi$ denote the multiplication operator defined by a function $\phi \in L^\infty(\TT)$.
\begin{definition} The Hankel operator corresponding to the symbol $\phi \in L^\infty(\TT)$ is defined as
\begin{equation*} H_\phi = \left(\Id - \Pi_\TT \right) \mathcal M_\phi \Pi_\TT : H^2(\TT) \to \left(H^2(\TT) \right)^\perp.\end{equation*}
\end{definition}
Let $[H_\phi] = (h_{k,l})_{k,l\ge 1}$ denote the matrix representing $H_\phi$ in the standard bases \begin{equation*}\mathcal B = \{z^{p-1} \ |\ p > 0\},\ \mathcal C = \{z^{-p} \ |\ p > 0\}\end{equation*}
of $H^2(\TT)$ and $H^2(\TT)^\perp$, respectively. Then
\begin{equation*} h_{k,l} = \langle \phi z^{l-1}, z^{-k} \rangle = \hat \phi(1-k-l).\end{equation*}
The truncated matrices associated with a symbol $\phi \in L^\infty(\TT)$ are relevant to us as well. For an infinite matrix $A = (a_{k,l})_{k,l \ge 1}$ denote $A_N = (a_{k,l})_{1 \le k,l \le N}$. We will require the following basic fact.  
\begin{lemma}\label{truncatednorm} $\lim_{N \to \infty} \Vert [H_\phi]_N \Vert_{\op} = \Vert H_\phi \Vert_{\op}$. \end{lemma}
Finally, let $E = \{z \in \TT \ |\ \Re z > 0\}$ denote the right half of the unit circle. Denote $H_E = H_{\mathbbm{1}_E}$, where $\mathbbm{1}_E$ is the indicator function of $E$. The Fourier coefficients of $\mathbbm{1}_E$ are specified by
\begin{equation}\label{fourierentries} \hat{\mathbbm{1}}_E(p) =\left\{\begin{array}{ll} \frac 1 2 & \text{if } p = 0,\\ \sin \left(\frac{\pi p} 2 \right)\frac 1 {\pi p} & \text{if } p \ne 0 \end{array}\right..\end{equation} Perhaps somewhat surprisingly, the operator $H_E$ is closely related to the commutators $C_n$. In particular, the proof of Theorem \ref{maintheorem} relies on the following.
\begin{lemma}\label{H_Enorm} $\Vert H_E \Vert_{\op} = \frac 1 2$. Hence by Lemma \ref{truncatednorm}, $\lim_{N \to \infty} \Vert [H_E]_N \Vert_{\op} = \frac 1 2$. \end{lemma}
We present the (simple) proof of the lemma in the following subsection. The inequality $\Vert H_E \Vert_{\op}\le \frac 1 2$ follows from the contents of Section \ref{mainproofsection}, though we include a separate proof using Nehari's Theorem on Hankel operators. The complementary inequality $\Vert H_E \Vert_{\op} \ge \frac 1 2$ is a direct consequence of Power's Theorem on Hankel operators with piecewise continuous symbols.
\subsection{On the norm of Hankel operators}\label{appendixb}
In this part, we apply two fundamental theorems on Hankel operators to obtain the proof of Lemma \ref{H_Enorm}. For a complex sequence $a = (a_k)_{k\in \NN}$, define the Hankel matrix $S_a = (a_{k+l})_{k,l\in \NN} : l^2(\NN) \to l^2(\NN)$.
\begin{theorem}[\cite{peller}] $S_a$ is bounded on $l^2(\NN)$ if and only if there exists $\phi \in L^\infty(\TT)$ such that $a_k = \hat \phi(k)$ for every $k \ge 0$. In this case,
\begin{equation*} \Vert S_a \Vert_{\op} = \inf\{\Vert \phi \Vert_{\infty} \ |\ \hat \phi(k) = a_k \text{ for every } k \ge 0\}. \end{equation*} \end{theorem}
We recall that
\begin{equation*} [H_E] = (h_{k,l})_{k,l\ge 1} = \left(\hat{\mathbbm{1}}_E(1-k-l)\right)_{k,l\ge 1},\end{equation*}
so the sequence associated with $[H_E]$ is
\begin{equation*} \left(\hat{\mathbbm{1}}_E(-1-k)\right)_{k\in \NN} = \left(\widehat{\bar z \mathbbm{1}_E}(k) \right)_{k \in \NN}.\end{equation*}
If we define
\begin{equation*} \phi(z) =\bar z \left(\mathbbm{1}_E(z) - \frac 1 2\right),\end{equation*}
and choose $k \ge 0$, then
\begin{equation*} \hat \phi(k) = \langle \mathbbm{1}_E-\frac 1 2, z^{k+1} \rangle = \hat{\mathbbm 1}_E(k+1) = \hat{\mathbbm{1}}_E(-1-k),\end{equation*}
therefore
\begin{conclusion} $\Vert H_E \Vert_{\op} \le \Vert \phi \Vert_{\infty} = \frac 1 2$.\end{conclusion}
To obtain the complementary inequality, assume that $\phi \in L^\infty(\TT)$ has well defined one-sided limits at every point of $\TT$. For $\alpha \in \TT$, define the jump of $\phi$ at $\alpha$ as
\begin{equation*} \phi_\alpha = \frac 1 2 \lim_{t \to 0^+} \left(\phi\left(\alpha e^{it} \right) - \phi \left(\alpha e^{-it} \right)\right). \end{equation*}
Then
\begin{theorem}[\cite{power}] The essential spectrum of the Hankel operator $H_\phi$ is given by
\begin{equation*} \sigma_{\text{ess}} \left(H_\phi \right) = [0, i \phi_1] \cup [0, i \phi_{-1}] \cup \left(\cup_{\alpha \in \TT \setminus\{\pm 1\}} \left[-\sqrt{-\phi_\alpha \phi_{\bar \alpha}}, \sqrt{-\phi_{\alpha} \phi_{\bar \alpha}} \right] \right).\end{equation*} \end{theorem}
The inequality $\Vert H_E \Vert_{\op} \ge \frac 1 2$ now follows, since
\begin{conclusion} The essential spectrum of $H_E$ equals $\left[-\frac 1 2, \frac 1 2 \right]$. Note that $\sigma_{\text{ess}}$ is a (closed) subset of the spectrum of $H_E$, hence $\Vert H_E \Vert_{\op} \ge \frac 1 2$. \end{conclusion}
\section{The proof of Theorem \ref{maintheorem}}\label{mainproofsection}
We begin with a few preliminary notations and definitions (\cite{biedenharnlouck}, \cite{varshalovich}). Recall that the spectrum of $J_x, J_y, J_z$ equals the set $\left\{j, j-1, ..., -j\right\}$, where $j = \frac 1 2 (n-1)$ is the spin number associated with the representation. Let
\begin{equation*} \mathcal E_{z,j} = \{e_m \ |\ m = j, j-1, ..., -j\}\end{equation*}
denote an orthonormal eigenbasis of $J_z$, such that $J_z e_m = m e_m$. The matrices representing $\mathbbm{1}_{(0,\infty)}(J_x),\ \mathbbm{1}_{(0,\infty)}(J_z)$ in $\mathcal E_{z,j}$ may be written as
\begin{equation*} P_{x,j} =\left(P_{x,j,m',m}\right)_{|m'|,|m|\le j} = \left(\begin{array}{cc} P_{1, x,j} & P_{2,x,j} \\ P_{2,x,j}^* & P_{3,x,j} \end{array}\right),\ P_{z,j} =  \left(\begin{array}{cc} \tilde I_{ j} & 0 \\ 0 & 0 \end{array}\right). \end{equation*}
Here, $\tilde I_j$ is the identity matrix of size $\left \lfloor j + \frac 1 2 \right \rfloor$. The matrix of $C_n$ is given by
\begin{equation*} \left[P_{x,j}, P_{z,j} \right] = \left(\begin{array}{cc} 0 & -P_{2,x,j} \\ P_{2,x,j}^* & 0 \end{array}\right),\end{equation*}
hence $\Vert C_n \Vert_{\op} = \Vert P_{2,x,j} \Vert_{\op}$. In our notations,
\begin{equation*} P_{2,x,j} = \left(P_{x,j,m',m}\right)_{ j\ge m' > 0 \ge m \ge -j }.\end{equation*}

We turn our attention to the "central elements" of $P_{x,j}$, that is, to the sequences $\left(P_{x,j+k,m',m} \right)_{k \in \NN}$ with $j,m',m$ fixed such that $m',m \in \{j,j-1,...,-j\}$. According to Conclusion \ref{pxjkmn},
\begin{equation*} \lim_{k \to \infty} P_{x,j+k,m',m} = \hat{\mathbbm{1}}_E(m-m'),\end{equation*}
where we recall that $\mathbbm{1}_E$ is the indicator function of the right half of the unit circle in $\CC$, as well as the symbol of the Hankel operator $H_E$ of Lemma \ref{H_Enorm}.

Thus, evidently, for fixed $N \in \NN$, the bottom left $N \times N$ corner of $P_{2,x,j}$ converges in $M_N(\CC)$ as $n \to \infty$. More precisely,
\begin{conclusion}\label{submatrices}Let $N \in \NN$, and assume that $j > N$. Let $C_{n,N} = (c_{n,k,l})_{k,l=1,...,N}$, where
\begin{equation*} c_{n,k,l} = \left\{\begin{array}{ll} P_{x,j,k, 1-l} & \text{if } n \in 2\NN +1 \\ P_{x,j,-\frac 1 2 + k, \frac 1 2 - l} & \text{if } n \in 2\NN\end{array}\right.\end{equation*}
Then $\lim_{n \to \infty} C_{n,N} = [H_E]_N$. \end{conclusion}
The proof of Theorem \ref{maintheorem} easily follows now, since $\lim_{N \to \infty} \Vert[H_E]_N \Vert_{\op} = \frac 1 2$ by Lemma \ref{H_Enorm}, and clearly
\begin{equation}\label{submatricesargument}  \frac 1 2 \ge \liminf_n \Vert C_n \Vert_{\op} \ge \liminf_n \Vert C_{n,N} \Vert_{\op} = \Vert [H_{E}]_N \Vert_{\op}.\end{equation}
Letting $N \to \infty$, we obtain the desired result.
\subsection{An extension of Theorem \ref{maintheorem}}\label{extensionsubsecsion}
In this part, we outline the proof of the following extension of Theorem \ref{maintheorem}.
\begin{theorem}\label{extensiontheorem} Fix $a \in [0,1)$ and $b \in (0,1]$, and let
\begin{equation*} C_{n,a,b} = \left[\mathbbm{1}_{\left(a\left(j+\frac 1 2 \right),\infty\right)}(J_x), \mathbbm{1}_{\left(0, b\left(j+\frac 1 2 \right)\right]}(J_z) \right].\end{equation*}
Then $\lim_{n \to \infty} \Vert C_{n,a,b} \Vert_{\op} = \frac 1 2$. \end{theorem}
This can be extended further by replacing $\mathbbm{1}_{\left(a\left(j + \frac 1 2 \right), \infty \right)}(J_x)$ with spectral projections corresponding to (not necessarily open) intervals whose end-points are $a_1 \left(j+\frac 1 2 \right), a_2 \left(j + \frac 1 2 \right)$, where $0 < a_1 < a_2 \le \infty$. However, once $0 < a_1$ is chosen, the projection arising from $J_z$ must correspond to an interval of the form above (or $(0,\infty)$), due to the limitations of the asymptotic estimate underlying the results of Section \ref{asymptoticsection}.

The proof of Theorem \ref{extensiontheorem} is rather identical to that of Theorem \ref{maintheorem}, except that we apply the more general Conclusion \ref{paxjkmn}, instead of Conclusion \ref{pxjkmn}. However, we note that the conjectured modulo $4$ dependence on the dimension $n$ seems to be (more or less) a unique feature of the case $a = b=0$ (see Figures \ref{notmod4}, \ref{againnotmod4}, for example).

Let $P_{x,a,j},\ Q_{z,b,j}$ denote the matrices representing the spectral projections in $\mathcal E_{z,j}$. Then
\begin{equation*} P_{x,a,j} = \left(P_{x,a,j,m',m} \right)_{|m'|,|m| \le j} = \left(\begin{array}{cc} P_{1,x, a,j} & P_{2,x, a,j} \\ P_{2,x,a,j}^* & P_{3,x,a,j} \end{array}\right),\ Q_{z,b,j} = \left(\begin{array}{cc} \tilde I_{b_j} & 0 \\ 0 & 0 \end{array}\right),\end{equation*}
where
\begin{equation*} \tilde I_{b_j} = \left(\begin{array}{cc} 0 & 0 \\ 0 & I_{b_j} \end{array}\right), \end{equation*}
and $b_j = \left \lfloor b \left(j + \frac 1 2 \right) \right \rfloor$, so that $\lim_{j \to \infty} b_j = \infty$. Hence, $C_{n,a,b}$ is represented by the matrix
\begin{equation*} [C_{n,a,b}] = \left(\begin{array}{cc} \left[P_{1,x,a,j}, \tilde I_{b_j} \right] & -\tilde I_{b_j} P_{2,x,a,j} \\ P_{2,x,a,j}^* \tilde I_{b_j} & 0 \end{array}\right).\end{equation*}
We may proceed to define a sequence of sub-matrices of $[C_{n,a,b}]$ as before, since $b_j \xrightarrow{j\to \infty} \infty$. According to Conclusion \ref{paxjkmn},
\begin{equation*} \lim_{k \to \infty} P_{a,x,j+k,m',m} = \hat{\mathbbm{1}}_{E_a}(m-m'),\end{equation*}
where $E_a = \left\{z \in \TT \ |\ \Re z > a \right\}$. It follows that the bottom left $N \times N$ corner of $\tilde I_{b_j} P_{2,x,a,j}$ converges in $M_N(\CC)$ to the truncated matrix of the Hankel operator $H_{\mathbbm{1}_{E_a}}$, whose norm equals $\frac 1 2$ by the same arguments that were applied to $H_E$.
\section{Matrix elements of spectral projections}\label{integralformulasection}
In this section, we establish a concise integral formula for the coefficients of $P_{x,j} = \left(P_{x,j,m',m}\right)_{|m'|,|m|\le j}$, which is the matrix representing the spectral projection $\mathbbm{1}_{(0,\infty)}(J_x)$ in the basis $\mathcal E_{z,j}$. The latter, we recall, is an eigenbasis of $J_z$. The eventual formula that we obtain for $P_{x,j,m',m}$ treats the cases $m-m' \in 2\ZZ$ and $m-m' \in 2\ZZ + 1$ separately.
\subsection{The Wigner small d-matrix}
The Wigner small d-matrix $d^j(\theta) = \left(d^j_{m',m}(\theta) \right)_{|m'|,|m| \le j}$ is the matrix of $e^{-i \theta J_y}$ in the basis $\mathcal E_{z,j}$. It is fundamental in the representation theory of $SU(2)$. The Wigner d-functions $d^j_{m',m}(\theta)$ are real valued $4\pi$-periodic trigonometric polynomials, commonly specified by the formula
\begin{gather*} d^j_{m',m}(\theta) = \sqrt{\frac{(j+m)!(j-m)!}{(j+m')!(j-m')!}} \cdot \\ \sum_{s} (-1)^{m'-m+s}{j+m' \choose j+m-s}{j-m' \choose s} \left(\cos \frac \theta 2 \right)^{2j+m-m'-2s} \left(\sin \frac \theta 2 \right)^{m'-m + 2s}.\end{gather*}

The parity of $d^j_{m',m}$ with respect to $m',m$ and $\theta$ is specified by (\cite{varshalovich}, 4.4)
\begin{equation}\label{djparity} d^j_{m',m}(-\theta) = (-1)^{m'-m}d^j_{m',m}(\theta) = d^j_{m,m'}(\theta) = d^j_{-m',-m}(\theta).\end{equation}
Another useful relation is
\begin{equation}\label{djpitranslation} d^j_{m',m}(\theta + \pi) = (-1)^{j-m} d^j_{m',-m}(\theta). \end{equation}
Finally, we will rely on the Fourier expansion of $d^j_{m',m}$, specified by (\cite{biedenharnlouck}, 3.78, \cite{fengwangyangjin})
\begin{equation}\label{djmnfourier} d^j_{m',m}(\theta) = e^{i\frac \pi 2 (m-m')} \sum_{\mu=-j}^j d^j_{m,\mu} \left(\frac \pi 2 \right) d^j_{m', \mu} \left( \frac \pi 2 \right) e^{-i \mu \theta}. \end{equation}
%
\subsection{Integral formula for $P_{x,j,m',m}$}
We may rotate one spin operator to another, and in particular, $J_x$ and $J_z$ are related by the formula $e^{i \frac \pi 2 J_y} J_x e^{-i \frac \pi 2 J_y} = J_z$.
This means that the vectors
\begin{equation*} f_m = e^{-i \frac \pi 2 J_y} e_m,\ m = j, j-1, ..., -j \end{equation*}
form an orthonormal eigenbasis of $J_x$, with $J_x f_m = m f_m$. Note that
\begin{equation*} P_{x,j,m',m} = \langle \mathbbm{1}_{(0,\infty)}(J_x) e_m, e_{m'} \rangle = \sum_{\mu > 0} \langle e_m, f_\mu \rangle \langle f_\mu, e_{m'} \rangle.\end{equation*}
Consequently,
\begin{equation}\label{pxjmninnerproductsformula}P_{x,j,m',m} = \sum_{\mu > 0} d^j_{m, \mu} \left( \frac \pi 2 \right) d^j_{m', \mu} \left( \frac \pi 2 \right).\end{equation}
Comparing the last expression with (\ref{djmnfourier}), we deduce that $P_{x,j,m',m}$ equals the sum of negative Fourier coefficients of $d^j_{m',m}$, up to multiplication by $e^{i \frac \pi 2 (m'-m)}$. Equivalently (see Section \ref{hankelssubsection}),
\begin{equation*} P_{x,j,m',m} = e^{i \frac \pi 2 (m'-m)}\left(\Id - \Pi_{\TT} \right)(d^j_{m',m})(0).\end{equation*}

$\Pi_{\TT}$ acts on $L^2(\TT)$ by $z^p \mapsto \mathbbm{1}_{[0,\infty)}(p) z^p$. A closely related operator is the periodic Hilbert transform $\HH_{\TT}$, which acts by $z^p \mapsto -i\sgn(p)z^p$. Thus, another equivalent formula is
\begin{equation}\label{pxjmnhilberttrans} P_{x,j,m',m} = \frac 1 2 e^{i \frac \pi 2 (m'-m)}\left( d^j_{m',m}(0) -\left\langle d^j_{m',m}, 1\right\rangle_{L^2(\TT)} - i \HH_{\TT}\left(d^j_{m',m}\right)(0) \right).\end{equation}
Furthermore, by (\ref{djmnfourier}), the zeroth Fourier coefficient of $d^j_{m',m}$ is specified by 
\begin{equation}\label{zerothfourier} \langle d^j_{m',m}, 1 \rangle_{L^2(\TT)} = \left\{\begin{array}{ll} e^{i \frac \pi 2 (m-m')} d^j_{m, 0}\left(\frac \pi 2 \right) d^j_{m',0} \left( \frac \pi 2 \right) & \text{if }j\in \NN \\ 0 & \text{if }j\in \frac 1 2 \NN \setminus \NN \end{array}\right..\end{equation}
$\HH_\TT$ maps even functions\footnote{i.e., functions $f \in L^2(\TT)$ with $\hat f(p) = \hat f(-p)$ for every $p \in \ZZ$.} to odd functions and vice versa. In light of the parity properties of (\ref{djparity}), we finally obtain the following.
\begin{conclusion}\label{djformulahilbert} The matrix elements of $\mathbbm{1}_{(0,\infty)}(J_x)$ in the eigenbasis $\mathcal E_{z,j}$ of $J_z$ are given by
\begin{equation*} P_{x,j,m',m} = \left\{\begin{array}{ll} \frac 1 2 e^{i \frac \pi 2 (m'-m)}\left(\delta_{m',m} - \langle d^j_{m',m}, 1 \rangle_{L^2(\TT)} \right) & \text{if } m'-m \in 2\ZZ \\ & \\ -\frac i 2 e^{i\frac \pi 2 (m'-m)} \HH_\TT(d^j_{m',m})(0) & \text{if } m'-m \in 2\ZZ + 1\end{array}\right. . \end{equation*}
Here $m',m \in \{j, j-1, ..., -j\}$ and $\delta_{m',m}$ is Kronecker's delta.\end{conclusion}

The periodic Hilbert transform admits the representation 
\begin{equation*} \HH_\TT f(\theta_0) = \frac 1 {4\pi} \lim_{\varepsilon \to 0^+} \int_{\varepsilon \le |\theta| \le 2\pi} f(\theta) \cot \left(\frac{\theta_0 - \theta} 4 \right) d\theta,\end{equation*}
where $f\in L^2(\TT)$ is a $4\pi$-periodic function. For $f = d^j_{m',m}$ with $m-m' \in 2\ZZ+1$, we obtain the formula
\begin{equation}\label{4pihilbtransform} \HH_\TT(d^j_{m',m})(0) =-\frac 1 {2\pi} \int_0^{2\pi} d^j_{m',m}(\theta) \cot \left( \frac \theta 4 \right) d\theta,\end{equation}
which will be studied in the next section.
\subsection{Matrix elements of $\mathbbm{1}_{\left(a\left(j+\frac 1 2 \right),\ \infty\right)}(J_x)$}
In this subsection, we extend (\ref{pxjmnhilberttrans}) to projections of the form $\mathbbm{1}_{\left(a\left(j+\frac 1 2 \right),\infty\right)}(J_x)$, where $0 \le a < 1$. More generally, the method may be used to obtain similar formulas for the elements of spectral projections corresponding to (not necessarily open) intervals with end-points at $a_1 \left(j + \frac 1 2\right),\ a_2 \left(j + \frac 1 2 \right)$, where $a_1 < a_2$ (hence also for projections corresponding to combinations of such intervals). The matrix elements in the present case are given by
\begin{gather*} P_{x,a,j,m',m} = \left \langle \mathbbm{1}_{\left(a \left(j + \frac 1 2 \right), \infty \right)}(J_x) e_m, e_{m'} \right \rangle = \sum_{\mu > a \left(j + \frac 1 2 \right)} \langle e_m, f_\mu \rangle \langle f_\mu, e_{m'} \rangle\\ = \sum_{\mu > a \left(j + \frac 1 2 \right)} d^j_{m, \mu} \left( \frac \pi 2 \right) d^j_{m',\mu} \left( \frac \pi 2 \right). \end{gather*}
This formula is the analogue of (\ref{pxjmninnerproductsformula}), and it also admits an interpretation through the Fourier expansions of the Wigner d-functions. Indeed, denote the Fourier coefficients of $d^j_{m',m}$ by $\hat d^j_{m',m}(p)$. Then
\begin{equation*} \hat d^j_{m',m}(p) = \left\{\begin{array}{ll} e^{-i \frac \pi 2(m'-m)} d^j_{m, -\frac p 2}\left(\frac \pi 2 \right) d^j_{m', -\frac p 2}\left(\frac \pi 2 \right) & \text{if } p = 2j, 2j-2, ..., -2j\\ 0 & \text{otherwise} \end{array}\right., \end{equation*}
therefore
\begin{equation*} P_{x,a,j,m',m} = e^{i \frac \pi 2 (m'-m)}\sum_{p< \lceil -a(2j+1)\rceil}\hat d^j_{m',m}(p).\end{equation*} 
Thus, $P_{x,a,j,m',m}$ equals the sum of Fourier coefficients of $d^j_{m',m}$ corresponding to indices lesser than $\lceil -a(2j+1) \rceil$. Equivalently, we may shift the Fourier expansion of $d^j_{m',m}$ using the relation $\widehat{z^{-p} f}(l) = \hat f(l+p)$ to obtain
\begin{equation*} P_{x,a,j,m',m} = e^{i \frac \pi 2 (m'-m)} \left(\Id - \Pi_{\TT}\right)\left(g^j_{m',m} \right)(0),\end{equation*}
where $g^j_{m',m}(\theta) = e^{-i a_j \frac \theta 2} d^j_{m',m}(\theta)$ and $a_j = \lceil -a(2j+1) \rceil$. As in (\ref{pxjmnhilberttrans}), we translate the former expression to
\begin{equation} P_{x,a,j,m',m} = \frac 1 2 e^{i \frac \pi 2(m'-m)} \left( \delta_{m',m} - \langle g^j_{m',m}, 1 \rangle_{L^2(\TT)} - i\HH_\TT(g^j_{m',m})(0) \right).\end{equation}
In the final part of the next section, we will study $\lim_{k \to \infty} P_{a,x,j+k,m',m}$.
\section{Limits of central matrix elements}\label{asymptoticsection}
In this section, we use an asymptotic approximation of Wigner d-functions by Bessel functions of the first kind in order to compute $\lim_{k \to \infty} \langle d^{j_k}_{m',m}, 1 \rangle_{L^2(\TT)}$ and $\lim_{k \to \infty} \HH_\TT \left(d^{j_k}_{m',m}\right)(0)$, where $j,m',m$ are fixed, $k \in \NN$ and $j_k = j+k$. Since $P_{x,j,m',m}$ are the elements of a symmetric matrix, we further assume that $m-m' \ge 0$. The values of the limits $\lim_{k \to\infty} P_{x,j_k,m',m}$ will then follow from Conclusion \ref{djformulahilbert}.

\subsection{Asymptotic approximation of Wigner d-functions}
The relevant asymptotic relation between Bessel functions and Wigner d-functions follows from a formula for the latter in terms of Jacobi polynomials.

Let $p \in \NN$. The Bessel function of the first kind $J_p$ may be specified by (\cite{abramowitzstegun})
\begin{equation} J_p(x)= \sum_{k=0}^\infty \frac{(-1)^k}{k!(k+p)!}\left( \frac x 2 \right)^{2k + p} = \frac 1 \pi \int_0^\pi \cos(pt - x\sin t)dt.\end{equation}
We note, for later use, that for $x \in \RR$, it holds that (\cite{abramowitzstegun}, 9.1.7, 9.2.1)
\begin{equation}\label{besselinfty} J_p(x) = \mathcal O(x^p),\ J_p(x) = \mathcal O\left(x^{-\frac 1 2} \right)\end{equation}
as $x \to 0$ and as $x \to +\infty$, respectively. The Bessel functions associated with negative integers are specified by (\cite{abramowitzstegun}, 9.1.5)
\begin{equation}\label{besselparity} J_{-p}(x) = (-1)^p J_p(x) = J_p(-x).\end{equation}

The Jacobi polynomials $P_{k}^{(\alpha,\beta)}$ are a class of classical orthogonal polynomials specified by (\cite{szego}, 4.3.1)
\begin{equation*} P_{k}^{(\alpha,\beta)}(x) = \frac{(-1)^{k}}{2^{k}k!}(1-x)^{-\alpha}(1+x)^{-\beta} \frac{d^{k}}{dx^{k}} \left[(1-x)^{\alpha}(1+x)^{\beta}(1-x^2)^{k}\right].\end{equation*}
They are orthogonal on the interval $[-1,1]$ with respect to the weight function $W^{(\alpha,\beta)}(x) = (1-x)^\alpha(1+x)^\beta$. 

The results of the present section are based on the following classical estimate. Assume that $\alpha > -1, \beta \in \RR$. Then (\cite{szego}, 8.21.12)
\begin{equation}\label{szegotheorem}\left(\sin \frac \theta 2\right)^{\alpha}\left(\cos \frac \theta 2 \right)^{\beta} P_k^{(\alpha,\beta)}(\cos \theta) = \frac{(k+\alpha)!}{r^\alpha k!} \sqrt{\frac \theta {\sin \theta}} J_k(r \theta) + E_k^{(\alpha, \beta)}(\theta),\end{equation}
where $r = k+\frac{\alpha+\beta+1}2$ and $E_k^{(\alpha, \beta)}(\theta) = \sqrt{\theta} \mathcal O(k^{-\frac 3 2})$ in intervals of the form $[0, \pi - \delta]$. Wigner d-functions are related to Jacobi polynomials by (\cite{biedenharnlouck}, 3.72)
\begin{equation*} d^j_{m',m}(\theta) = \sqrt{\frac{(j+m)!(j-m)!}{(j-m')!(j+m')!}} \left(\sin \frac \theta 2\right)^{m-m'}\left(\cos \frac \theta 2 \right)^{m+m'} P_{j-m}^{(m-m', m+m')}(\cos \theta).\end{equation*}
Thus, choosing $k = j-m,\ \alpha = m-m',\ \beta = m+m'$, we obtain a powerful asymptotic approximation\footnote{We refer the reader to  \cite{rowedeguisesanders} for a survey of the asymptotic properties of Wigner d-functions.} of $d^j_{m',m}$.
\begin{conclusion}\label{wignerdasymptotics} Fix $m,m' \in \NN$ or $m,m' \in \frac 1 2 \NN \setminus \NN$ such that $m-m' \ge -1$. Then
\begin{equation*} d^j_{m',m}(\theta) = C_{j,m',m}\sqrt{\frac \theta {\sin \theta}} J_{m-m'} \left(\frac{2j+1}{2} \theta \right) + E_{j-m}^{(m-m', m+m')}(\theta), \end{equation*}
where $E_{j-m}^{(m-m', m+m')} = \sqrt\theta \mathcal O\left(j^{- \frac 3 2}\right)$ in intervals of the form $[0, \pi - \delta]$, and
\begin{equation*} C_{j,m',m} = \sqrt{\frac{(j-m')!(j+m)!}{(j-m)!(j+m')!}} \frac 1 {\left(j + \frac 1 2 \right)^{m-m'}} \end{equation*} satisfies $\lim_{j\to \infty} C_{j, m',m} = 1$.\end{conclusion}
The asymptotic approximation may be extended to $d^j_{m',m}$ with $m-m'<0$ using the parity relations (\ref{djparity}), and to further intervals using (\ref{djpitranslation}).
\subsection{Integrals involving Wigner d-functions}
Conclusion \ref{wignerdasymptotics} together with the fact that $\lim_{x \to \infty} J_{m-m'}(x) = 0$ by (\ref{besselinfty}) imply that $\lim_{j \to \infty} d^{j}_{m',m}(\theta) = 0$ for $m-m' \ge 0$ and $\theta \in (0,\pi)$ fixed. This remains true when $m-m' < 0$, as may be shown using the parity properties of $d^j_{m',m}$. In particular, if $j \in \NN$ and $m'$ is fixed, then $\lim_{k \to \infty} d^{j+k}_{m',0} \left( \frac \pi 2 \right) = 0$. Thus, in light of (\ref{zerothfourier}), we find that
\begin{equation} \lim_{k \to \infty} \langle d^{j+k}_{m',m}, 1 \rangle_{L^2(\TT)} = 0.\end{equation}
The latter also follows from the next lemma, which will be used in the analysis of $\HH_\TT(d^j_{m',m})(0)$. As before, denote $j_k = j+k$.
\begin{lemma} Let $f\in L^\infty(\TT)$. Then $\lim_{k \to \infty} \langle d^{j_k}_{m',m}, f \rangle_{L^2(\TT)} = 0$.\end{lemma}
\begin{proof}
Recall that $d^j_{m',m}(\theta)$ is an element of a unitary matrix, so $|d^j_{m',m}(\theta)|\le 1$ for every $j,m',m$ and $\theta$. Moreover, $\lim_{k \to \infty} d^{j_{k}}_{m',m}(\theta) = 0$ for $\theta \in (0,\pi)$. By the dominated convergence theorem, it follows that
\begin{equation*} \lim_{k \to \infty} \int_0^\pi f(\theta) d^{j_k}_{m',m}(\theta) d\theta = 0\end{equation*}
for every $f \in L^\infty(\TT)$. The symmetries (\ref{djpitranslation}), (\ref{djparity}) of $d^j_{m',m}(\theta)$ imply, similarly, that the integrals over the intervals $[\pi, 2\pi]$ and $[-2\pi, 0]$ converge to $0$.\end{proof}
The lemma is not immediately applicable to $\HH_\TT(d^j_{m',m})(0)$, since by (\ref{4pihilbtransform}),
\begin{equation*}\HH_\TT(d^j_{m',m})(0)=  -\frac 1 {2\pi} \int_{0}^{2\pi} d^j_{m',m}(\theta) \cot \left(\frac \theta 4 \right) d\theta,\end{equation*}
and $\cot \left(\frac \theta 4 \right)$ is unbounded. However, it allows us to truncate this integral to an interval of the form $[0, \delta]$, where $\delta > 0$ is arbitrarily small.
\begin{conclusion} Fix $0 < \delta < 1$. Then $\lim_{k \to \infty} \left[\HH_\TT(d^{j_k}_{m',m})(0)-I_{j_k, \delta} \right] = 0$, where
\begin{equation*} I_{j, \delta} = -\frac 1 {2\pi} \int_0^\delta d^j_{m',m}(\theta) \cot \left( \frac \theta 4 \right) d\theta.\end{equation*}
\end{conclusion}
At this point, we wish to use the asymptotic formula of Conclusion \ref{wignerdasymptotics}. The error satisfies $E_{j-m}^{(m-m', m+m')}(\theta) = \sqrt \theta \mathcal O\left(j^{-\frac 3 2}\right)$, and the function $\sqrt \theta \cot \frac \theta 4$ is integrable, hence
\begin{equation*} \lim_{k \to \infty} \int_0^{\delta} E_{j_k-m}^{(m-m', m+m')}(\theta) \cot \left( \frac \theta 4 \right) d\theta = 0.\end{equation*}
Therefore, we have obtained that
\begin{multline}\label{finallimit} \lim_{k \to \infty} \HH_\TT(d^{j_k}_{m',m})(0) =\\ - \frac 1 {2\pi}\lim_{k \to \infty} \int_0^\delta \sqrt{\frac \theta {\sin \theta}} J_{m-m'} \left(\frac{2j_k+1}{2} \theta \right) \cot\left( \frac \theta 4 \right) d\theta.\end{multline}
Let $\HH_{\RR} : L^2(\RR) \to L^2(\RR)$ denote the standard Hilbert transform, specified by
\begin{equation*} \HH_\RR f(x) = -\frac 1 \pi \lim_{\varepsilon \to 0^+} \int_\varepsilon^\infty \frac{f(x+t) - f(x-t)} t dt.\end{equation*}
As in the case of $\HH_\TT$, if $f \in L^2(\RR)$ is even then $\HH_\RR(f)$ is odd, and vice versa. Finally, we are ready to prove the main result of this subsection.
\begin{claim} Let $j_k = j+k$ as above, with $j \in \frac 1 2 \NN$ fixed and $k \in \NN$, and fix $m',m \in \{j, j-1, ..., -j\}$. Then
\begin{equation*} \lim_{k \to \infty} \HH_\TT \left(d^{j_k}_{m',m} \right)(0) = \HH_\RR \left(J_{m-m'} \right)(0).\end{equation*} When $m-m' \in 2\ZZ$, this simply says that $\HH_\TT \left(d^{j_k}_{m',m}\right)(0) = 0 = \HH_\RR \left(J_{m-m'} \right)(0)$.\end{claim}
\begin{proof}
Assume that $m-m' \in 2\ZZ + 1$. By the substitution $x = \left(j+\frac 1 2 \right) \theta$ in (\ref{finallimit}), it suffices to establish that $\lim_{k \to \infty} I_{j_k} = \HH_{\RR}\left(J_{m-m'} \right)(0)$, where
\begin{equation*} I_j = -\frac 1 {2\pi} \int_0^{\left(j + \frac 1 2 \right) \delta } \sqrt{\frac{\frac x {j+\frac 1 2}}{ \sin\left(\frac x {j+\frac 1 2} \right)}} J_{m-m'}(x) \cot \left(\frac x {4\left(j + \frac 1 2 \right)}\right) \frac{dx}{j+\frac 1 2}\end{equation*}
with $j \in \frac 1 2 \NN$. To this end, denote
\begin{gather*} f^{(1)}_j(x) = \sqrt{\frac{\frac x {j+\frac 1 2 }}{\sin \left(\frac x {j+\frac 1 2}\right)}},\ f^{(2)}_j(x) = \frac 1 {j+\frac 1 2 } \cot \left( \frac x {4\left( j + \frac 1 2 \right) } \right),\\ f_j=\mathbbm{1}_{\left(0,\delta \left( j + \frac 1 2 \right) \right)}f^{(1)}_jf^{(2)}_j  J_{m-m'},\end{gather*}
Then $\lim_{j \to \infty} f_j^{(1)}(x) = 1$ and $\lim_{j \to \infty} f_j^{(2)}(x) = \frac 4 x$. Additionally,
\begin{equation*} 0 < f^{(1)}_j(x) < M_\delta = \max_{\theta \in [0,\delta]} \sqrt{\frac \theta {\sin \theta}},\ 0 < f_j^{(2)}(x) \le \frac 4 x\end{equation*}
for every $j \in \frac 1 2 \NN,\ x \in \left[0, \delta \left(j + \frac 1 2 \right) \right]$. Thus,
\begin{equation*}|f_j(x)| \le \frac 4 x \left| J_{m-m'}(x) \right|,\ \lim_{j \to \infty} f_j(x) = \frac 4 x J_{m-m'}(x)\end{equation*}
for every $x \in (0,\infty)$. Finally, $\mathbbm{1}_{(0,\infty)}(x) \frac 4 x J_{m-m'}(x) \in L^1(\RR)$, hence by the dominated convergence theorem, we deduce that
\begin{equation*}-\frac 1 {2\pi} \lim_{j \to \infty} \int_{\RR} f_j(x) dx= -\frac 2 \pi \int_0^\infty \frac{J_{m-m'}(x)} x dx = \HH_\RR(J_{m-m'})(0),\end{equation*}
where the last equality holds since $J_{m-m'}$ is an odd function by (\ref{besselparity}).
\end{proof}
\subsection{Application to $P_{x,j,m',m}$}
The results of the previous subsection, together with Conclusion \ref{djformulahilbert}, imply that
\begin{equation*} \lim_{k\to \infty} P_{x,j_k,m',m} = \left\{\begin{array}{ll} \frac 1 2  & \text{if } m-m' = 0,\\ -\frac{i }{2} e^{-i \frac \pi 2 (m-m')}\HH_\RR(J_{m-m'})(0) & \text{if } m-m'\ne 0 \end{array}\right.,\end{equation*}
where we recall that $\HH_\RR (J_{p})(0) = 0$ for $p$ even. The Hilbert transform of the Bessel function $J_p$ admits the alternative representation (\cite{poularikas}, 15.9.1)
\begin{equation*} \HH_{\RR}(J_p)(x) =\frac 1 \pi \int_0^\pi \sin(x\sin t - pt)dt.\end{equation*}
Thus, when $p \ne 0$, we find that $\HH_\RR(J_p)(0) = -\frac 1 {\pi p}\left(1-(-1)^{p}\right)$. Additionally, $i e^{-i \frac \pi 2 p} = \sin \left(\frac \pi 2 p \right)$ whenever $p$ is odd, therefore finally
\begin{equation*} \lim_{k \to \infty} P_{x,j_k,m',m} = \left\{\begin{array}{ll} \frac 1 2 & \text{if } m-m' =0,\\ \sin \left(\frac \pi 2(m-m')\right) \frac 1 {\pi(m-m')} & \text{if } m-m' \ne 0. \end{array}\right. \end{equation*}
In light of (\ref{fourierentries}), this establishes the relation between the elements of $\mathbbm{1}_{(0,\infty)}(J_x)$ and the Fourier coefficients of $\mathbbm{1}_E$, where we recall that $E$ denotes the right half of the unit circle in $\CC$.
\begin{conclusion}\label{pxjkmn} $\lim_{k \to \infty} P_{x,j_k,m',m} = \hat{\mathbbm{1}}_E(m-m')$.\end{conclusion}
\subsection{Limits of central elements of $\mathbbm{1}_{\left(\left(j+\frac 1 2 \right)a ,\ \infty \right)}(J_x)$}
In the final part of the previous section, we obtained a formula for the matrix elements of the projection $\mathbbm{1}_{\left(a\left(j + \frac 1 2 \right), \infty \right)}(J_x)$ in the basis $\mathcal E_{z,j}$, with $0 \le a < 1$. Namely, for
\begin{equation*} g^j_{m',m}(\theta) = e^{-i a_j \frac \theta 2} d^j_{m',m}(\theta),\ a_j = \lceil - a(2j+1) \rceil,\end{equation*}
we saw that
\begin{equation*} P_{x,a,j,m',m} =  \frac 1 2 e^{i \frac \pi 2 (m'-m)} \left(\delta_{m',m} - \langle g^j_{m',m}, 1 \rangle_{L^2(\TT)} - i \HH_{\TT} \left( g^j_{m',m} \right)(0)\right).\end{equation*}

Notably, $| g^j_{m',m}| = | d^j_{m',m}|$, therefore most of the arguments in the analysis of $P_{x,j,m',m} = P_{x,0,j,m',m}$ remain valid for $P_{x,a,j,m',m}$ with $a>0$. Specifically,
\begin{equation*} g^j_{m',m}(0) = d^j_{m',m}(0) = \delta_{m',m},\ \lim_{k \to \infty} \langle g^{j_k}_{m',m}, 1 \rangle_{L^2(\TT)} = 0,\end{equation*}
and since $\lim_{j \to \infty} \frac{a_j}{2j+1 } = -a$, 
we can also establish that
\begin{equation*} \lim_{k \to \infty} \HH_{\TT} \left(g^{j_k}_{m',m} \right)(0) =\HH_{\RR}(f_{a,m-m'})(0),\end{equation*}
where $f_{a,p}(x) = e^{ai x} J_{p}(x)$. Moreover, using the parity of $\cos(ax), \sin(ax)$ and $J_p(x)$, we see that
\begin{equation*} \HH_\RR(f_{a,p})(0) =  \left\{\begin{array}{ll}- \frac{2i} \pi \int_0^\infty \frac{\sin(ax) J_p(x)} x dx & \text{if } p \in 2\ZZ,\\ & \\ -\frac 2 \pi \int_0^\infty \frac{\cos(ax) J_p(x)} x dx & \text{if } p \in 2\ZZ+1. \end{array}\right.\end{equation*}
The Bessel function $J_p$ is part of the integral kernel of the $p$th order Hankel transform, which provides a straightforward way to evaluate the integrals above. For $a < 1$, we have that (\cite{bateman}, 8.2.33, 8.7.2, 8.7.27)
\begin{equation*} \int_0^\infty \frac{\sin(ax) J_0(x)}{x} dx = \sin^{-1}(a) \end{equation*}
and otherwise when $p \in 2\ZZ \setminus \{0\}$,
\begin{equation*} \int_0^\infty \frac{\sin(ax) J_p(x)} x dx = \frac 1 p \sin\left(p \sin^{-1}(a)\right). \end{equation*}
Similarly, if $p \in 2\ZZ + 1$,
\begin{equation*} \int_0^\infty \frac{\cos(ax) J_p(x)}x dx = \frac 1 p \cos\left(p \sin^{-1}(a) \right). \end{equation*}
Combining the above, we obtain a generalization of Conclusion \ref{pxjkmn}.
\begin{conclusion}\label{paxjkmn} Assume that $a = \cos \alpha = \sin \left(\frac \pi 2 - \alpha \right)$, with $\alpha \in [0, \frac \pi 2 )$. Then, using trigonometric identities for angle difference, we obtain that
\begin{gather*} \lim_{k \to \infty} P_{x,a, j_k, m',m} = \left\{\begin{array}{ll} \frac \alpha \pi & \text{if } m-m' = 0 \\ \frac 1 {\pi(m-m')}\sin\left((m-m')\alpha \right) & \text{if } m-m' \ne 0 \end{array}\right. = \hat{\mathbbm 1}_{E_a}(m-m'),\end{gather*}
where $E_a = \left\{z \in \TT \ |\ \Re z > a \right\}$.\end{conclusion}
\section{Miscellaneous proofs}\label{miscproofs}
\subsection{Proof of Theorem \ref{linecommutator}}
Recall that $X, \Xi$ act on a smooth function $f \in L^2(\RR)$ by
\begin{equation*} Xf (x) = xf(x),\ \Xi f(x) = -i\hbar f'(x).\end{equation*}
Let $\sigma_\hbar$ denote the rescaling $f(x) \mapsto \sqrt \hbar f(\hbar x)$. Then $\sigma_{\hbar}$ is unitary with respect to the inner product
\begin{equation*} \langle f, g \rangle  = \int_{-\infty}^\infty f(x) \bar g(x) dx,\end{equation*}
since
\begin{equation*} \langle \sigma_{\hbar} f, \sigma_\hbar g \rangle = \hbar \int_{-\infty} f(\hbar x) \bar g(\hbar x) dx = \langle f, g \rangle. \end{equation*}
Let $\mathcal F$ denote the (unitary) Fourier transform on $L^2(\RR)$, acting on a function $f \in L^1(\RR) \cap L^2(\RR)$ by
\begin{equation*} \mathcal F f(\xi) = \frac 1 {\sqrt {2 \pi}} \int_{-\infty}^\infty f(x) e^{-ix\xi} dx.\end{equation*} %
We define the semiclassical Fourier transform $\FF_\hbar$, using the scaling properties of $\FF$, as
\begin{equation*} \mathcal F_\hbar = \sigma_{\hbar^{-1}} \mathcal F = \mathcal F \sigma_{\hbar}.\end{equation*}

The observables $X, \Xi$ are conjugated by $\FF_\hbar$, that is,
\begin{equation*} \Xi = \FF_\hbar^{-1} X \FF_{\hbar}. \end{equation*}
Consequently, so are the functional calculi of $\Xi, X$, where the latter consists of multiplication operators. Recall that
\begin{equation*} \Pi_X = \mathbbm{1}_{(0,\infty)}(X) = \mathcal M_{\heaviside},\ \Pi_\Xi = \mathbbm{1}_{(0,\infty)}(\Xi),\end{equation*}
hence
\begin{conclusion} The projections $\Pi_X, \Pi_\Xi$ are related by
\begin{equation*} \Pi_\Xi = \FF^{-1} \sigma_{\hbar} \mathcal M_{\mathbbm{1}_{(0,\infty)}} \sigma_{\hbar^{-1}} \mathcal F = \mathcal F^{-1} \mathcal M_{\mathbbm{1}_{(0,\infty)}} \mathcal F,\end{equation*}
where $\mathcal M_f$ denotes the operator of multiplication by $f$. \end{conclusion}

In particular, $\Pi_\Xi$ is independent of $\hbar$, hence
\begin{equation*} C_\hbar^{(1)} = \left[\Pi_X, \Pi_\Xi\right] = C^{(1)} \end{equation*}
for some fixed, bounded operator $C^{(1)}$ on $L^2(\RR)$. Next, we recall that the Hardy space on $\RR$ is given by
\begin{equation*} H^2(\RR) = \{f\in L^2(\RR) \ |\ \FF f(\xi) = 0 \ \text{for every } \xi<0\},\end{equation*}
therefore $\Pi_\Xi = \Pi_\RR$ is the Cauchy-Szeg\"{o} projection on $H^2(\RR)$, and consequently
\begin{equation*} C^{(1)} = \left[ \mathcal M_{\mathbbm{1}_{(0,\infty)}}, \Pi_{\RR} \right].\end{equation*}
Let $C(z) = \frac{z-i}{z+i}$ denote the Cayley transform (which maps $(0,\infty) \subset \RR$ onto $\{\Im z < 0\} \subset \TT$). The unitary operator $U_C : L^2(\TT) \to L^2(\RR)$ specified by
\begin{equation*} U_C f(x) = \pi^{-\frac 1 2 }(x+i)^{-1} f\left(C(x)\right)\end{equation*}
is known (\cite{rosenblum}, p.92) to map $H^2(\TT)$ onto $H^2(\RR)$. 
\begin{lemma} $U_C^* \Pi_X U_C = \mathcal M_{\mathbbm 1_{\{\Im z < 0\}}}$ and $U_C^* \Pi_\Xi U_C = \Pi_\TT$, where the latter denotes the Cauchy-Szeg\"{o} projection on $H^2(\TT)$. \end{lemma}
\begin{proof}
Note that
\begin{equation*} U_C^*\psi(z) = 2i \frac{\sqrt \pi}{1-z} \psi \left( C^{-1}(z)\right).\end{equation*}
For a bounded function $\psi : \RR \to \RR$ and $f \in L^2(\TT)$, we obtain that
\begin{equation*} U_C^* \mathcal M_\psi U_C f(z) = 2i \frac{\sqrt \pi}{1-z} \psi\left(C^{-1}(z)\right) \cdot \left(U_C f \right)\left(C^{-1}(z) \right) = \left( \psi \circ C^{-1}\right)(z) f(z).\end{equation*}
Then,
\begin{equation*} \mathbbm{1}_{(0,\infty)} \circ C^{-1} = \mathbbm{1}_{C\left((0,\infty)\right)} = \mathbbm{1}_{\{\Im z < 0\}},\end{equation*}
which means that $U_C^* \Pi_X U_C = \mathcal M_{\mathbbm 1_{\{\Im z < 0 \}}}$.

Next, $U_C$ is unitary and maps $H^2(\TT)$ onto $H^2(\RR)$, hence it maps $H^2(\TT)^\perp$ onto $H^2(\RR)^\perp$. It follows immediately that $U_C^* \Pi_\RR U_C = \Pi_\TT$. \end{proof}
In light of the previous lemma, we conclude that
\begin{equation*} \left \Vert C^{(1)} \right\Vert_{\op} = \left \Vert U_C^* C^{(1)} U_C \right \Vert_{\op} = \left \Vert \left[\mathcal M_{\mathbbm{1}_{\{\Im z < 0\}}}, \Pi_\TT \right] \right \Vert_{\op}. \end{equation*}
Consider the translation operator $\tau$ on $L^2(\TT)$, specified by $f(z) \mapsto f\left(e^{i \frac \pi 2} z \right)$. Then $\tau z^m = e^{i \frac {\pi m}{2}} z^m$ for every $m \in \ZZ$, therefore $\tau^* \Pi_\TT \tau = \Pi_\TT$. At the same time,
\begin{equation*} \tau^* \mathcal M_f \tau = \mathcal M_{ \tau^*f},\end{equation*}
therefore
\begin{equation*} \tau^* \mathcal M_{\mathbbm{1}_{\{\Im z < 0\}}} \tau = \mathcal M_{\mathbbm{1}_E},\end{equation*}
where we recall that $E = \{z \in \TT \ | \ \Re z > 0\}$. Finally, if we denote $\Pi_\TT^\perp = \Id - \Pi_\TT$, then
\begin{gather*} \left[\mathcal M_f, \Pi_\TT \right] = \left(\Pi_\TT + \Pi_\TT^\perp \right) \left[ \mathcal M_f, \Pi_\TT \right] \left( \Pi_\TT + \Pi_\TT^\perp\right)\\ = \Pi_\TT^\perp \mathcal M_f \Pi_\TT - \Pi_\TT \mathcal M_f \Pi_\TT^\perp = H_f - H_f^*,\end{gather*}
where $H_f$ denotes the Hankel operator with symbol $f$.
\begin{conclusion} $\left \Vert \left[ \mathcal M_{\mathbbm{1}_{\{\Im z < 0\}}}, \Pi_\TT \right] \right \Vert_{\op} = \left \Vert \left[ \mathcal M_{\mathbbm{1}_E}, \Pi_\TT \right] \right \Vert_{\op}$,
where
\begin{equation}\label{basicestcase}\left[ \mathcal M_{\mathbbm 1_E}, \Pi_\TT \right] = H_E \oplus \left(-H_E \right)^* : H^2(\TT) \oplus H^2(\TT)^\perp \to H^2(\TT)^\perp \oplus H^2(\TT).\end{equation} It follows from Lemma \ref{H_Enorm} that $\Vert \left[ \Pi_X, \Pi_\Xi \right] \Vert_{\op} = \frac 1 2$. \end{conclusion}
\subsection{Proof of Theorem \ref{ringcommutator}}
Recall that we have defined the operators $\Theta, Z$ on $L^2(\TT)\simeq L^2\left([0,2\pi),\frac 1 {2\pi} d\theta\right)$ by
\begin{equation*} \Theta u(\theta) = \theta u (\theta),\ Zu(\theta) = -i \frac{2\pi} n u'(\theta),\end{equation*}
where $u\in C^\infty(\TT)$ and $n \in \NN$. We are interested in
$C^{(2)}_n = \left[\Pi_\Theta, \Pi_Z \right]$,
where
\begin{gather*} \Pi_\Theta = \heaviside\left(\cos \Theta \right) = \mathcal M_{\mathbbm 1_E},\\ \Pi_Z = \heaviside \left(\cos Z \right).\end{gather*}

The proof that $\lim_{n\to \infty} \Vert C^{(n)} \Vert_{\op} = \frac 1 2$ immediately reduces to Lemma \ref{H_Enorm}, since $\{z^k \ |\ k \in \ZZ\}$ is an eigenbasis of $Z$ (analogous to $\mathcal E_{z,j}$ for $SU(2)$), with
\begin{equation*} Z (z^k) = \frac{2\pi k}{n} z^k.\end{equation*}
This means that
\begin{equation*} \Pi_Z (z^k) = \mathbbm 1_E\left(\lambda_{k,n}\right) z^k,\end{equation*}
where $\lambda_{k,n} = e^{\frac{2\pi k}{n} i}$.
The matrix elements of $\Pi_\Theta$ are specified by
\begin{equation*} \langle \Pi_\Theta z^l, z^k \rangle = \hat{\mathbbm{1}}_E(k-l).\end{equation*}

Consequently, the matrix elements of $C^{(\text{2})}_{n}$ are specified by
\begin{gather*} c^{(2)}_{n,k,l} = \langle C^{(2)}_n z^l, z^k \rangle = \langle \Pi_\Theta \Pi_Z z^l, z^k \rangle - \langle \Pi_Z \Pi_\Theta z^l, z^k \rangle\\ = \mathbbm{1}_E\left(\lambda_{l,n} \right) \langle \mathbbm 1_E z^l, z^k \rangle - \langle \mathbbm 1_E z^l, \Pi_Z z^k \rangle = \left(\mathbbm{1}_E(\lambda_{l,n}) - \mathbbm{1}_E(\lambda_{k,n})\right) \hat{\mathbbm{1}}_E(k-l),\end{gather*}
In particular, when $\frac n 4 < l < \frac {3n} 4$ and $0 \le k < \frac n 4$,
\begin{equation*} c^{(2)}_{n,k,l}= -\hat{\mathbbm{1}}_E(k-l).\end{equation*}
\begin{conclusion} Fix some positive $N \in \NN$, and assume that $n > 4N$. Define
\begin{equation*} C_{n,N}^{(2)} = \left(a_{k,l}\right)_{k,l = 1,...,N} = \left(c^{(2)}_{n,\lceil \frac n 4 \rceil - k,\lceil \frac n 4 \rceil + l-1} \right)_{k,l = 1,...,N}.\end{equation*}
Then $a_{k,l} = -\hat{\mathbbm 1}_E(1-k-l)$. It follows that
\begin{equation*} -C_{n,N}^{(2)} = \left[H_E\right]_N \end{equation*}
is the truncated Hankel matrix associated with $H_E$. By Lemma \ref{H_Enorm} and the same argument as in (\ref{submatricesargument}), we deduce that $\lim_{n \to \infty} \Vert C^{(2)}_n \Vert_{\op} = \frac 1 2$.\end{conclusion}
If we replace $\mathbbm 1_{(0,\infty)}$ with $\mathbbm 1_{(a, \infty)}$, where $a \in (0,1)$, then the proof remains largely unchanged, except for the use of the Hankel operator $H_{\mathbbm 1_{E_a}}$ (as in Conclusion \ref{paxjkmn}) instead of $H_E$.
\subsection{Proof of Theorem \ref{heisclaim}}
The proof of Theorem \ref{heisclaim} may be obtained by straightforward computations. However, we will use a geometric model as follows.

We identify the standard basis of $V_n = l^2(\ZZ_n)$ with
\begin{equation*}  \Delta_n = \left\{ \delta_{\frac {2\pi k} n} \ |\ k = 0,1,...,n-1\right\} = \left\{\delta_{\frac {2\pi k} n} \ |\ k \in \ZZ \right\},\end{equation*}
where $\delta_{\frac {2\pi k} n}$ is the Dirac measure supported in $\frac {2\pi k} n \in \ZZ_n \subset \TT \simeq [0,2\pi)$. For a vector $v = \sum_{k=0}^{n-1} v_k \delta_{\frac{2\pi k}{n}}$ and a bounded, measurable function $f : \TT \to \CC$ we will use the notation
\begin{equation*} f v = \sum_{k = 0}^{n-1} f\left(\frac{2\pi k}{n} \right) v_k \delta_{\frac{2\pi k}{n}},\end{equation*}
and refer to the operator $v \mapsto f v$ as the multiplication operator $\mathcal M_f : V_n \to V_n$.

Given $f : \TT \to \CC$, define the discretization
\begin{equation*} A_n(f) = f A_n(1) = \frac 1 {\sqrt { n}} \sum_{k = 0}^{n-1} f\left(\frac {2\pi k} n  \right) \delta_{\frac {2\pi k} n}.\end{equation*}
The representation of $H(\ZZ_n)$ is realized on $\left(V_n, \langle \cdot, \cdot \rangle_n \right)$, where $\langle \cdot, \cdot \rangle_n$ is specified by $\langle \delta_{\frac {2\pi k} n}, \delta_{\frac {2\pi l} n} \rangle_n = \delta_{kl}$. In these settings, $g_2$ is the multiplication operator $\mathcal M_z$, and $g_1$ is the operator of translation by $\frac {2\pi} n$.

In particular, $g_1 \delta_{\frac{2\pi k} n} = \delta_{\frac{2\pi (k-1)}{n}}$, and we note that
\begin{equation*}  g_1 A_n(f) = \frac 1 {\sqrt n} \sum_{k=0}^{n-1} f\left(\frac{2\pi k}{n} \right) \delta_{\frac{2\pi (k-1)} n} = \frac 1 {\sqrt n} \sum_{k = 0}^{n-1} f \left(\frac{2\pi(k+1)} n \right) \delta_{\frac{2\pi k} n } = A_n(\tau_n f),\end{equation*}
where $\tau_n f(\theta) = f\left(\theta + \frac{2\pi} n\right)$. Thus,
\begin{equation*} g_1 A_n(z^k) = e^{2\pi \frac k n i} A_n(z^k) = \lambda_{k,n} A_n(z^k),\end{equation*}
therefore
\begin{equation*} \mathcal E_n = \left\{e_{k,n} \ |\ k = 0,1,...,n-1\right\} = \left\{A_n\left(z^k \right)\ |\ k \in \ZZ \right\}\end{equation*}
is an eigenbasis of $g_1$, orthonormal with respect to $\langle \cdot, \cdot \rangle_n$ (as may be seen by a straightforward calculation). $\Delta_n$ is clearly an orthonormal eigenbasis of $g_2$.

Let $\FF_n$ denote the (unitary) discrete Fourier transform, specified by
\begin{equation*} \langle \FF_n v, \delta_{\frac {2\pi k}{n}} \rangle_n = \frac 1 {\sqrt n}\sum_{l=0}^{n-1} v_l e^{-\frac{2\pi k l}{n}i}.\end{equation*}
Then
\begin{lemma} $g_1 = \FF_n^{-1} g_2 \FF_n$.\end{lemma}
\begin{proof}
Note that
\begin{equation*} \langle \FF_n \delta_{\frac{2\pi m} n}, \delta_{\frac{2\pi k} n} \rangle_n = \frac 1 {\sqrt n} e^{-2\pi \frac{km}n i},\end{equation*}
hence
\begin{equation*}  \FF_n \delta_{\frac{2\pi m} n} = \frac 1 {\sqrt n} \sum_{ k = 0}^{n-1} \left(e^{\frac{2\pi k} n i} \right)^{-m} \delta_{\frac{2\pi k}{n}} = A_n(z^{-m}),\end{equation*}
which means that
\begin{equation*} g_2 \FF_n \delta_{\frac{2\pi m}{n}} = \frac 1 {\sqrt n} \sum_{k = 0}^{n-1} \left(e^{\frac{2\pi k} n} \right)^{-(m-1)} \delta_{\frac{2\pi k}{n}} = \FF_n \delta_{\frac{2\pi(m-1)} n}.\end{equation*}
We conclude that
\begin{equation*} \FF_n^{-1} g_2 \FF_n \delta_{\frac{2\pi m}{n}} = \delta_{\frac{2\pi (m-1)}{n}} = g_1 \delta_{\frac{2\pi m} n},\end{equation*}
therefore $g_1 = \FF_n^{-1} g_2 \FF_n$. \end{proof}
%

We have that $\Pi_2 = \mathbbm 1_E\left(\mathcal M_z \right)=\mathcal M_{\mathbbm{1}_E}$, therefore $\Pi_1 = \FF_n^{-1} \mathcal M_{\mathbbm 1_E} \FF_n$.
Since $\mathcal E_n$ is an eigenbasis of $g_1$, it holds that
\begin{equation*} \Pi_1 e_{k,n} = \mathbbm{1}_E(\lambda_{k,n})e_{k,n}.\end{equation*}
The matrix elements of $\Pi_2$ in $\mathcal E_n$ are given by
\begin{gather*} \langle \Pi_2 e_{l,n}, e_{k,n} \rangle_n = \langle  A_n\left(\mathbbm{1}_E z^l \right), A_n\left(z^k \right)\rangle_n \\= \langle z^l A_n(\mathbbm 1_E), A_n(z^{k})\rangle_n = \langle A_n(\mathbbm 1_E), A_n(z^{k-l}) \rangle_n.\end{gather*}
Here, we have used the fact that $\mathcal M_{f_1} A_n(f_2) = f_1 A_n(f_2) = A_n(f_1 f_2)$ and that
\begin{equation*} \langle f A_n(f_1), A_n( f_2) \rangle_n = \langle A_n(f_1), \bar f A_n( f_2) \rangle_n.\end{equation*}

The proof of Theorem \ref{heisclaim} reduces to Lemma \ref{H_Enorm}, as in all previous cases. We demonstrate this using $\mathcal E_n$ (though $\Delta_n$ works just as well). 
The matrix elements of the commutator $C^{(3)}_n = \left[\Pi_1, \Pi_2 \right]$ are specified by
\begin{equation*} c^{(3)}_{n,k,l} = \langle C^{(3)}_n e_{l,n}, e_{k,n} \rangle_n = \left(\mathbbm 1_E(\lambda_{l,n}) - \mathbbm 1_E(\lambda_{k,n})\right) \langle \Pi_2 e_{l,n}, e_{k,n} \rangle_n.\end{equation*}
In particular, when $\frac n 4 < l < \frac {3n} 4$ and $0 \le k < \frac n 4$,
\begin{equation*}  c^{(3)}_{n,k,l}= \langle A_n(\mathbbm 1_E), A_n(z^{k-l}) \rangle_n.\end{equation*}
If $f_1, f_2 \in L^2(\TT)$, then
\begin{equation*} \langle A_n(f_1), A_n(f_2) \rangle_n = \frac 1 {2\pi} \sum_{k = 0}^{n-1} \left[ f_1\left(\frac {2\pi k} n \right) \bar f_2 \left(\frac {2\pi k} n \right)\frac{2\pi}{n} \right] \xrightarrow{n \to \infty} \langle f_1,f_2 \rangle_{L^2(\TT)}.\end{equation*}
Thus,
\begin{conclusion} Fix some positive $N \in \NN$, and assume that $n > 4N$. Define
\begin{equation*} C_{n,N}^{(3)} = \left(b_{n,k,l}\right)_{k,l = 1,...,N} = \left(c^{(3)}_{n,\lceil \frac n 4 \rceil - k,\lceil \frac n 4 \rceil + l-1} \right)_{k,l = 1,...,N}.\end{equation*}
Then $\lim_{n\to \infty} b_{n,k,l} = \lim_{n \to \infty} \langle A_n(\mathbbm 1_E), A_n(z^{1-k-l}) \rangle_n = \hat{\mathbbm 1}_E(1-k-l)$. It follows that
\begin{equation*} \lim_{n \to \infty} C_{n,N}^{(3)} = \left[H_E\right]_N \end{equation*}
is the truncated Hankel matrix associated with $H_E$. By Lemma \ref{H_Enorm} and the same argument as in (\ref{submatricesargument}), we deduce that $\lim_{n \to \infty} \Vert C^{(3)}_n \Vert_{\op} = \frac 1 2$.\end{conclusion}
As in the case of Theorem \ref{ringcommutator}, if we replace $\mathbbm 1_{(0,\infty)}$ with $\mathbbm 1_{(a, \infty)}$, where $a \in (0,1)$, then the proof remains largely unchanged, except for the use of the Hankel operator $H_{\mathbbm 1_{E_a}}$ (as in Conclusion \ref{paxjkmn}) instead of $H_E$.
\subsection{Proof of Theorem \ref{euclideancommutator}}
The proof is immediate. Indeed, $\Pi_\Phi(f) = \Pi_\TT f - \hat f(0)$, and
\begin{equation*} \Pi_{X_1} = \mathcal M_{\mathbbm{1}_{(0,\infty)}(R \cos \phi)} = \mathcal M_{\mathbbm{1}_E},\end{equation*}
therefore
\begin{equation*} C^{(4)}_R = \left[\mathcal M_{\mathbbm{1}_E}, \Pi_\TT \right],\end{equation*}
and $\Vert C^{(4)}_R \Vert_{\op} = \frac 1 2$, as we have already seen in (\ref{basicestcase}).
\section{Discussion and a general conjecture}\label{morecases}
We begin with an informal interpretation of Theorem \ref{maintheorem}, based on a realization of the representations of $SU(2)$ through Berezin-Toeplitz quantization of the unit sphere $S^2 \subset \mathbb R^3$. This will lead us to formulate a conjectured, generalized version of Theorem \ref{maintheorem}, using the language of quantization. Subsequently, we will explore the conjectured formulation in a number of concrete examples.

In what follows, $\mathcal L(\HH)$ denotes the space of self-adjoint operators on a finite dimensional Hilbert space $\HH$. Let $(M, \omega)$ denote a closed\footnote{i.e., compact and without boundary.}, quantizable\footnote{i.e., $\frac \omega {2\pi}$ represents an integral de-Rham cohomology class.} symplectic manifold. A Berezin-Toeplitz quantization (\cite{charles, schlichen, lefloch}) of $M$ produces a sequence of finite dimensional complex Hilbert spaces $\left(\HH_\hbar \right)_{\hbar \in \Lambda}$, where $0$ is an accumulation point of $\Lambda \subset (0,\infty)$ and $\lim_{\hbar \to 0^+} \dim \HH_\hbar = + \infty$, together with surjective linear maps $T_\hbar : C^\infty(M) \to \mathcal L\left(\HH_\hbar\right)$, such that
\begin{enumerate}
\item{$T_\hbar(1) = \Id_{\HH_\hbar}$,}
\item{if $f \ge 0$, then $T_\hbar(f) \ge 0$,}
\item{$\Vert f \Vert_\infty - O(\hbar) \le \left\Vert T_\hbar(f) \right \Vert_{\op} \le \Vert f \Vert_{\infty}$,}
\item \label{correspprinc} {$\left\Vert \frac i \hbar \left[T_\hbar(f), T_\hbar(g) \right] - T_\hbar\left(\{f,g\}\right)\right\Vert_{\op} = O(\hbar)$,}
\item{$\left\Vert T_\hbar\left(f^2\right) - T_\hbar(f) \right\Vert_{\op} = O(\hbar)$}
\end{enumerate}
for every $f, g \in C^\infty(M)$. Here $\Vert f \Vert_\infty = \max_M |f|$ is the uniform norm and $\{f, g\}$ is the Poisson bracket of $f,g$. The existence of a Berezin-Toeplitz quantization in these rather general settings is a non-trivial fact, though if $(M, \omega)$ is a closed K\"{a}hler manifold, then the construction itself is quite direct. Item \ref{correspprinc} above is known as the \textit{correspondence principle}, and it is central to our interpretation.

Let us identify $S^2$ with the complex projective space $\CC P^1$ via the stereographic projection through the north pole, and let $\rho$ denote the standard action of $SU(2)$ on $\CC P^1$, given by
\begin{equation*} \rho(U)\left([z_1 : z_2 ] \right) = [\alpha z_1 - \bar \beta z_2 : \beta z_1 + \bar \alpha z_2 ],\ U = \left(\begin{array}{cc} \alpha & -\bar \beta \\ \beta & \bar \alpha \end{array}\right) \in SU(2). \end{equation*}
In addition to the properties specified above, the Berezin-Toeplitz quantization of $S^2 \simeq \CC P^1$ is $SU(2)$-equivariant, meaning that $\HH_\hbar$ carries an irreducible, unitary representation $\rho_\hbar$ of $SU(2)$ such that
\begin{equation*} T_\hbar \left(f \circ \rho(U)^{-1}\right) = \rho_\hbar(U) T_\hbar(f) \rho_\hbar(U)^*\end{equation*}
for every $\hbar \in \Lambda$, $f \in C^\infty\left(\CC P^1 \right)$ and $U \in SU(2)$. Here, $\hbar^{-1} = n = \dim \HH_\hbar $, and
$\Lambda = \left\{\left. n^{-1} \ \right| \ n = 1, 2, ...\right\}$.
The spin operators $J_x, J_y, J_z \in \LL\left(\HH_\hbar \right)$ are then, up to normalization, the quantum counterparts of the Cartesian coordinate functions $x, y, z : \CC P^1 \to \RR$. Specifically,
\begin{equation*} T_\hbar(x) = \frac 1 {n+1} J_x,\ T_\hbar(y) = \frac 1 {n+1} J_y,\ T_\hbar(z) = \frac 1 {n+1} J_z.\end{equation*}
Since $\mathbbm{1}_{(0,\infty)}$ is unaffected by positive rescalings, Theorem \ref{maintheorem} means that
\begin{equation*}\lim_{n \to \infty} \Vert C_n \Vert_{\op} = \lim_{\hbar \to 0^+} \left \Vert \left[ \mathbbm{1}_{(0,\infty)}\left(T_\hbar(x)\right), \mathbbm{1}_{(0,\infty)}\left(T_\hbar(z) \right) \right] \right\Vert_{\op} = \frac 1 2.\end{equation*}

Finally, our loose interpretation of this result goes as follows. We consider the spectral projections
\begin{equation*}\mathbbm{1}_{(0,\infty)}(J_x),\ \mathbbm{1}_{(0,\infty)}(J_z)\end{equation*}
as a pair of observables that are somehow related (\cite{zz1, zz2, zz3}) to the indicator functions of the hemispheres $\{x > 0\}, \{z > 0\} \subset S^2$. Thus, we interpret Theorem \ref{maintheorem} as an informal attempt to explore the correspondence principle (item \ref{correspprinc} above) in the context of discontinuous classical observables\footnote{To the best of our knowledge, a well-defined, useful (in the context of quantization) notion of Poisson bracket which is applicable to discontinuous observables does not exist.}. At the moment, it is unclear whether $C_n$ corresponds to a well-defined classical object as $n \to \infty$. Still, the behavior of $\left(C_n \right)_{n \ge 2}$ appears to be related to the intersection of the boundaries of the respective hemispheres, that is, to the points $\pm(0,1,0)$.

To see this, note that $\HH_\hbar$ may be identified with the space of homogeneous polynomials of degree $n -1$ in two complex variables, such that $\rho_\hbar$ becomes the standard irreducible unitary representation of $SU(2)$ in the latter space. Assume that $v_n \in \HH_\hbar$ is a polynomial which realizes the norm of $C_n$, i.e., assume that $\Vert C_n v_n \Vert = \Vert C_n \Vert_{\op}$. Our numerical simulations suggest that $v_n$ concentrates about the points $\pm(0,1,0)$ when $n \to \infty$, as illustrated in the following images.

\begin{figure}[H]

    \begin{minipage}{0.5\textwidth}
        \centering
        \includegraphics[width=1\textwidth]{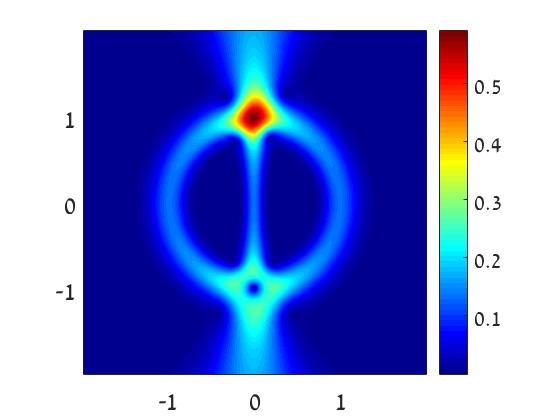} 
    \end{minipage}\hfill
    \begin{minipage}{0.5\textwidth}
        \centering
        \includegraphics[width=1\textwidth]{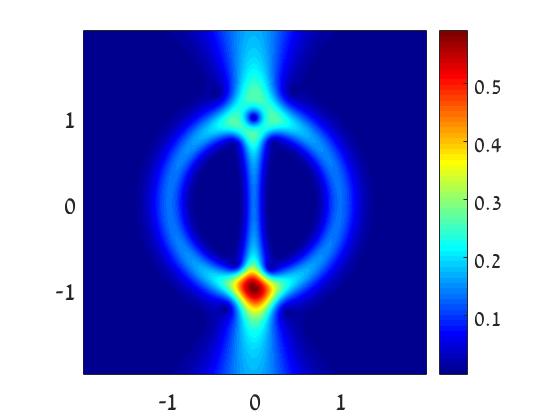} 
    \end{minipage}
    \caption{(originally by Y. Le Floch) The modulus of (unit) eigenvectors of $C_{101}$ corresponding to extremal eigenvalues, realized as polynomials on $\CC$.}
\end{figure}

\begin{figure}[H]

    \begin{minipage}{0.5\textwidth}
        \centering
        \includegraphics[width=1\textwidth]{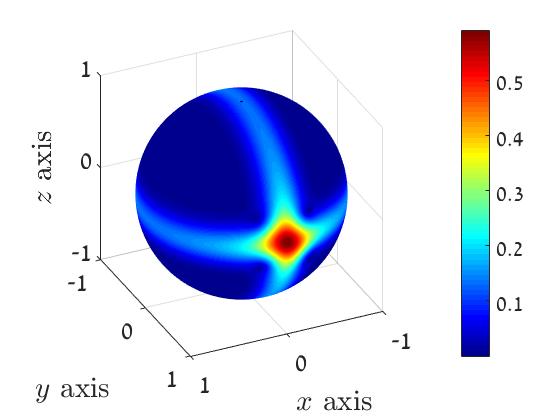} 
    \end{minipage}\hfill
    \begin{minipage}{0.5\textwidth}
        \centering
        \includegraphics[width=1\textwidth]{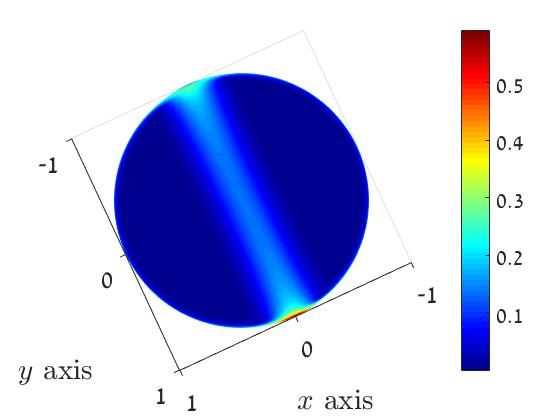} 
    \end{minipage}
    \caption{The image above to the left, reproduced with the eigenvector realized as a function on $S^2$ using the stereographic projection.}
\end{figure}
%

More generally, assume that $T_\hbar(f), T_\hbar(g)$ are a pair of quantum observables arising from smooth, non-commuting observables $f,g$ on some quantizable phase space $M$, and let $I, J \subset \RR$ denote some intervals. As before, we view the projections 
\begin{equation*} \Pi_{\hbar, f, I} = \mathbbm{1}_I\left(T_\hbar(f) \right),\ \Pi_{\hbar, g, J} = \mathbbm{1}_J\left(T_\hbar(g) \right) \end{equation*}
as a pair of observables that are related to the domains $f^{-1}(I),\ g^{-1}(J) \subset M$. The numerical evidence (Figures \ref{notmod4}, \ref{againnotmod4} in particular) and the results presented in this work appear to support the following conjecture, which was inspired by recent findings pertaining to quantization of domains in phase space (\cite{pol1, pol2, charpol}).
\begin{conjecturee} If $M$ is $2$-dimensional, and if the intersection of the boundaries of the domains $f^{-1}(I),\ g^{-1}(J)$ is non-empty and transversal, then  $\Pi_{I, f, \hbar}, \Pi_{J, g, \hbar}$ are maximally non-commuting in the semiclassical limit, i.e.,
\begin{equation*} \lim_{\hbar \to 0^+} \left \Vert \left[\Pi_{I, f, \hbar}, \Pi_{J, g, \hbar} \right] \right\Vert_{\op} = \frac 1 2.\end{equation*}\end{conjecturee}

In the context of $S^2$, Theorem \ref{extensiontheorem} is a modest extension of our main result, and agrees with the conjecture. Similarly, consider the sequence
\begin{equation*} C_{n,a} = \left[ \mathbbm{1}_{\left(a\left(j+\frac 1 2 \right), \infty \right)}(J_x), \mathbbm{1}_{\left(a\left(j + \frac 1 2 \right), \infty \right)}(J_z) \right],\end{equation*}
where $a \in [0, 1)$. A rigorous calculation of $\lim_{n \to \infty} \Vert C_{n,a} \Vert_{\op}$ for $a > 0$ is not possible using our current method (due to limitations in the applicability of the asymptotic estimate (\ref{szegotheorem})). However, the conjecture forecasts a transition in the behavior of $\left( \Vert C_{n,a} \Vert_{\op} \right)_{n \ge 2}$ as $a$ crosses the value $\frac 1 {\sqrt 2}$. According to our numerical simulations, this indeed seems to be the case. The following images are the analogues of Figure \ref{su2} above for the choices $a = 0.25,\ 0.75$ and $a = 0.7,\ 0.71$.
\begin{figure}[H]
        \centering
        \includegraphics[width=1.3\textwidth,center]{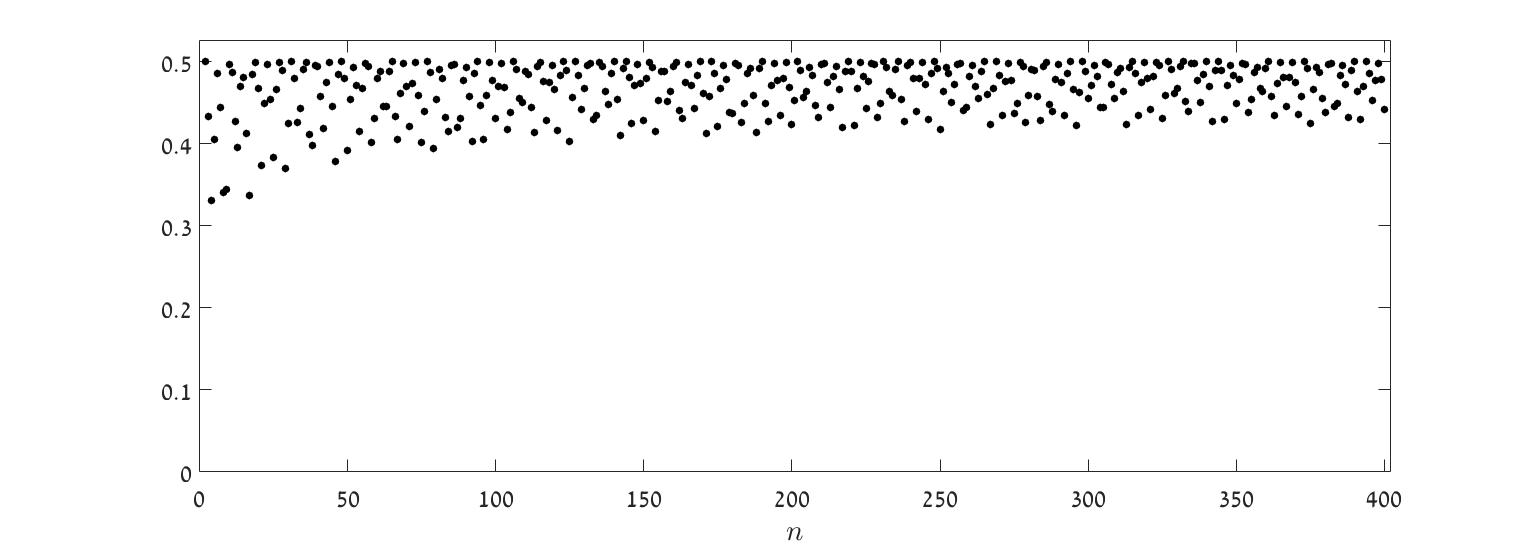} 
        \includegraphics[width=1.3\textwidth,center]{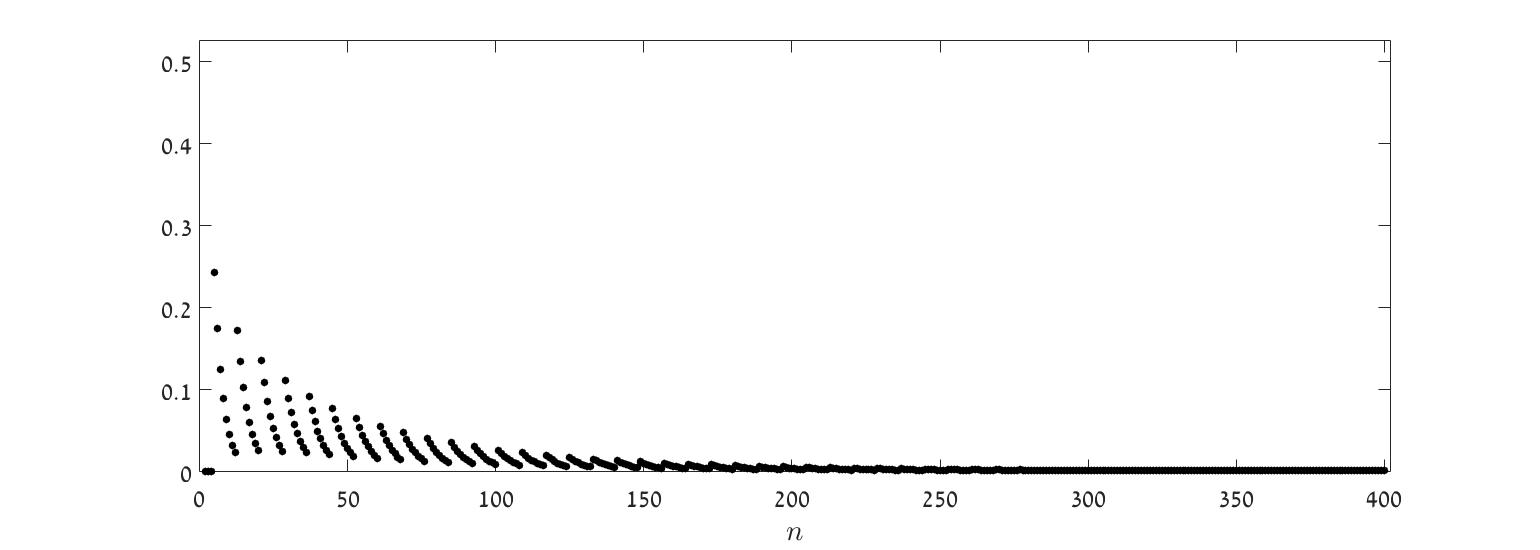} 
    \caption{A plot of $\Vert C_{n,a} \Vert_{\op}$ as a function of $n$, where $a = 0.25$ (top), $a = 0.75$ (bottom).}\label{notmod4}
\end{figure}
\begin{figure}[H]
        \centering
        \includegraphics[width=1.3\textwidth,center]{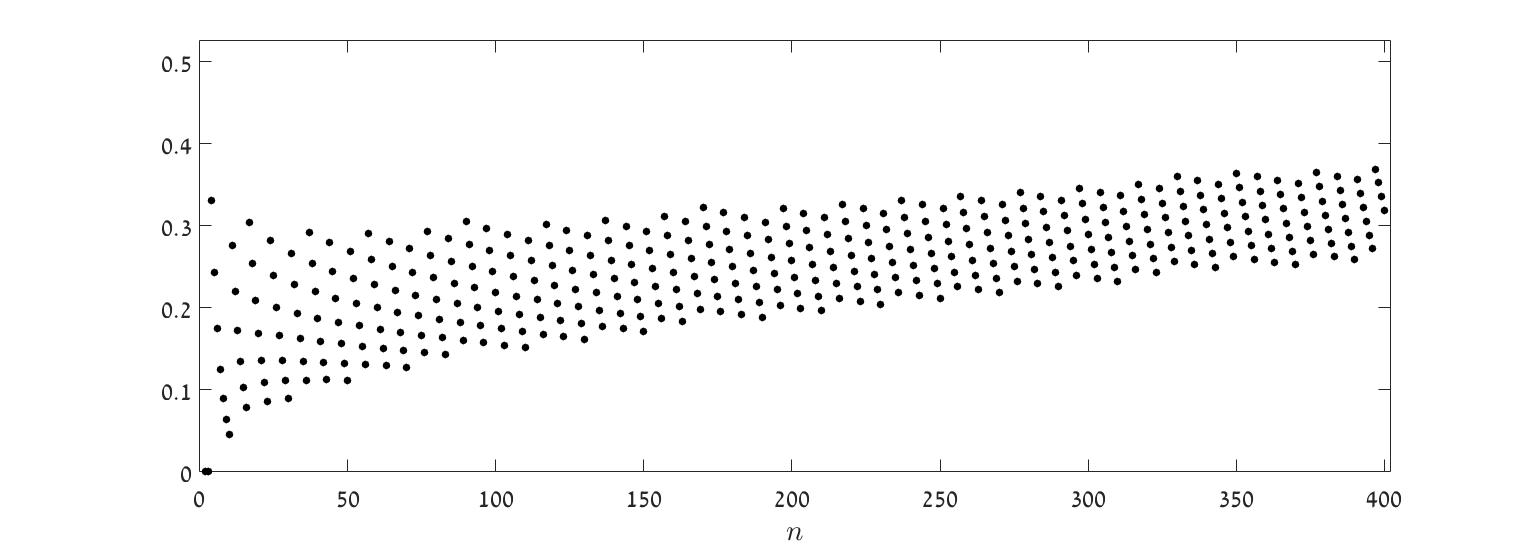} 
        \includegraphics[width=1.3\textwidth,center]{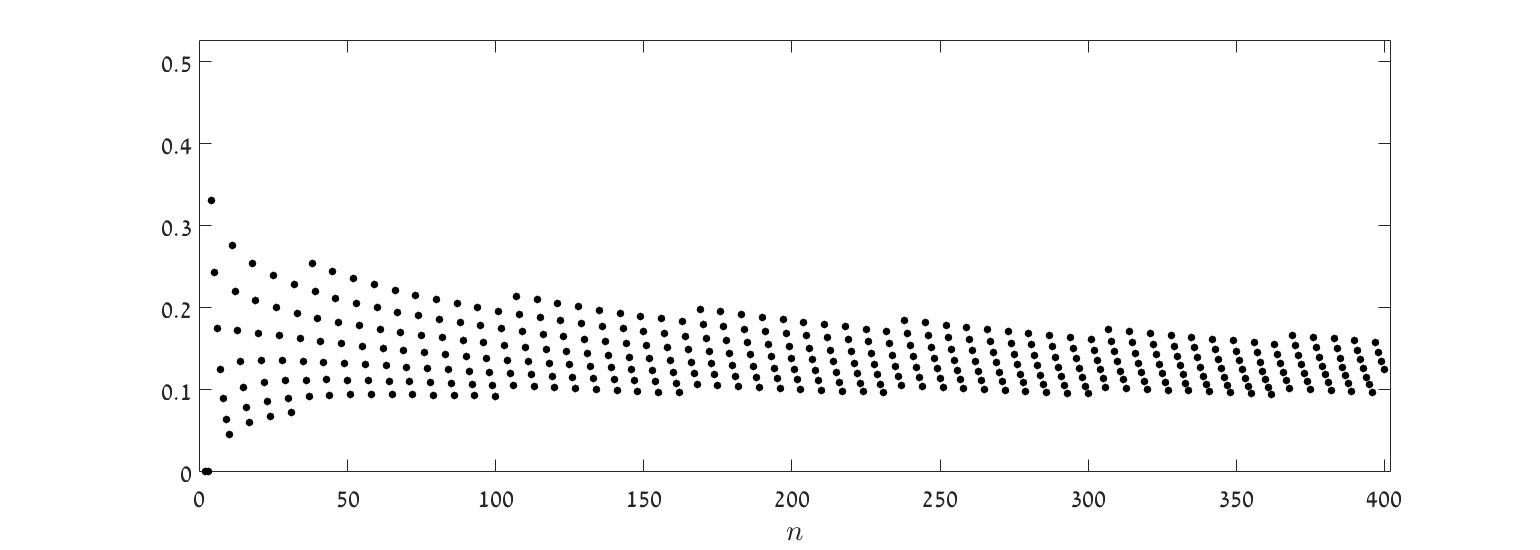} 
    \caption{A plot of $\Vert C_{n,a} \Vert_{\op}$ as a function of $n$, where $a = 0.7$ (top), $a = 0.71$ (bottom). Compare, also, with the image for $a = 0.75$.}\label{againnotmod4}
\end{figure}
Although we formulated the conjecture using the specific notion of Berezin-Toeplitz quantization, we expect it to hold in the context of similar or standard quantization schemes as well (Weyl quantization, in particular).

\renewcommand{\abstractname}{Acknowledgements}
\begin{abstract}
This research has been partially supported by the European Research Council Advanced Grant 338809, and by the European Research Council Starting Grant 757585. I wish to express my sincere gratitude to the European Research Council.\hfill

I wish to thank my supervisor Leonid Polterovich for his suggestion to pursue the questions addressed in this manuscript, for his interest, commitment, and insightful guidance throughout the project, and for the many useful ideas, remarks and corrections. I also wish to thank Yohann Le Floch for numerous contributions to this project throughout the years, and for his patience, help and encouragement. The results presented here could not have been obtained without the contributions and assistance of both.

I wish to thank my co-supervisor Lev Buhovsky for his support, for many clarifications, and for his kindly dedication, during our meetings, to carefully review the details of certain key parts of this work. I wish to thank Mikhail Sodin for his involvement during key stages of the project and for his clear, illuminating advice and input.

Finally, I wish to thank David Kazhdan for invaluable conversations and for his suggestion to consider the finite Heisenberg groups, prior to which the problem seemed entirely intractable.
\end{abstract}

\end{document}